\documentclass[12pt,english]{article}
\usepackage{times}
\usepackage[T1]{fontenc}
\usepackage[authoryear]{natbib}
\usepackage{amssymb,amsmath,amsthm}
\usepackage{mathtools}
\usepackage{comment}
\usepackage{graphicx}
\usepackage[pdfborder={0 0 0},colorlinks=true,
citecolor=blue,
linkcolor=blue,
anchorcolor=blue,
urlcolor=blue]{hyperref}
\usepackage{chngpage}
\usepackage{rotating}
\usepackage{pdflscape}
\usepackage{fullpage}
\usepackage{adjustbox}
\usepackage{setspace}
\usepackage{csquotes}
\usepackage{accents}
\usepackage{tipx}
\usepackage{enumitem}
\usepackage{lscape}
\usepackage{booktabs}

\usepackage[small,compact]{titlesec}
\usepackage{rotate}
\usepackage{caption}
\usepackage{subcaption}
\usepackage[table]{xcolor}
\usepackage{multirow}
\usepackage{arydshln}
\usepackage{amssymb}
\usepackage[titletoc,toc,title,page,header]{appendix}
\usepackage{minitoc}

\noptcrule

\usepackage{bibunits}
\defaultbibliographystyle{apalike2}
\defaultbibliography{literature}

\usepackage{tikz}
\usepackage{epigraph}

\usetikzlibrary{shapes.geometric, fit,calc,positioning,decorations.pathreplacing,matrix, patterns}
\usepackage{xr}
\newcolumntype{C}[1]{>{\centering\arraybackslash}m{#1}}

\theoremstyle{plain}        
\theoremstyle{plain}        
\theoremstyle{plain}		\newtheorem{prop}{Proposition}
\theoremstyle{plain}		
\theoremstyle{plain}		
\theoremstyle{plain}        \newtheorem{cor}{Corollary}
\theoremstyle{plain}		\newtheorem{remark}{Remark}
\theoremstyle{plain}		\newtheorem{definition}{Definition}
\theoremstyle{plain}        \newtheorem{ass}{Assumption}
\theoremstyle{plain}        \newtheorem{example}{Example}
\theoremstyle{plain}        \newtheorem{defn}{Definition}

\theoremstyle{plain}        \newtheorem{corollary}{Corollary}

\usepackage{xr}

\bibpunct[: ]{(}{)}{;}{a}{,}{,}

\begin{document}
    \tikzstyle{VertexStyle} = [shape            = ellipse,
                               minimum width    = 6ex,%
                               draw]

    \tikzstyle{EdgeStyle}   = [->,>=stealth']

\title{Heterogeneous Treatment Effects and\\Causal Mechanisms\thanks{We thank Scott Abramson, Neal Beck, Matt Blackwell, Andrew Little, Molly Offer-Westort, and Scott Tyson; seminar audiences at New York University, Princeton, and Berkeley; and participants at the NYU Abu Dhabi Theory in Methods Workshop and PolMeth XL for helpful feedback.}}
\date{\today}
\author{Jiawei Fu\footnote{Assistant Professor, Duke University. \url{jiawei.fu@duke.edu}} \and Tara Slough\footnote{Associate Professor, New York University. \url{tara.slough@nyu.edu}}}

\maketitle
\thispagestyle{empty}
\begin{abstract}
The credibility revolution advances the use of research designs that permit identification and estimation of causal effects. However, understanding which mechanisms produce measured causal effects remains a challenge. The dominant current approach to the quantitative evaluation of mechanisms relies on the detection of heterogeneous treatment effects (HTEs) with respect to pre-treatment covariates. This paper develops a framework to understand when the existence of such heterogeneous treatment effects can support inferences about the activation of a mechanism. We show first that this design cannot provide evidence of mechanism activation without additional, generally implicit, exclusion assumptions. Further, even when these assumptions are satisfied, the presence of HTEs supports the inference that mechanism is active but the absence of HTEs is generally uninformative about mechanism activation. We provide novel guidance for interpretation and research design in light of these findings.

\end{abstract}

\thispagestyle{empty}

\clearpage
\doparttoc % Tell to minitoc to generate a toc for the parts
\faketableofcontents % Run a fake tableofcontents command for the partocs

\setcounter{page}{1}
\doublespacing

\begin{bibunit}
Social scientists often make claims about \emph{how} or \emph{why} a treatment affects a given outcome. Researchers might tell us, for example, not only that an outreach program increased COVID-19 reporting rates, but that it did so \emph{by} building trust in government. Or researchers might claim that access to local media increases knowledge of local politics because of the content of news coverage, \emph{not} because of local political advertising. Claims of this type are often made on the basis of heterogeneous treatment effects (HTEs): the effect of the outreach program on COVID-19 reporting is larger in places that initially distrusted government \citep{haimetal2021}, for example, or the effect of local media on knowledge of local politics does not change over the election cycle \citep{moscowitz2021}. In fact, a majority of recent articles that report HTEs interpret them as evidence about \emph{how} or \emph{why} a treatment affects a given outcome.

\label{comment:theory}But the theoretical foundation for making such claims lags far behind empirical practice. While we have strong theoretical foundations for making claims about treatment effects themselves, theory provides little guidance about the relationship between HTEs and tests of mechanisms. We provide that guidance, developing a theoretical framework that clarifies the assumptions that researchers need to make in order to use HTEs as evidence for or against potential mechanisms. Our analysis reveals that current practice is often misleading---but also that, as is so often the case, making implicit assumptions explicit can correct such errors in future work.

This framework is important because political scientists so often use HTEs to make claims about causal mechanisms. Surveying the 2021 volumes of three leading journals in political science---the \emph{American Journal of Political Science} (\emph{AJPS}), the \emph{American Political Science Review} (\emph{APSR}), and the \emph{Journal of Politics} (\emph{JoP})---we find that a majority (56\%) of quantitative studies estimate HTEs and that, conditional on reporting any HTEs, the vast majority of articles (82\%) interpret them as providing information about mechanisms. Taken together, these figures indicate that almost half (46\%) of recent quantitative empirical articles in these journals use HTEs to assess mechanisms.\footnote{Table \ref{tab:class2} further documents that the share of studies that use of HTE for mechanism detection is similar across all quantitative research designs/identification strategies in common usage.} Beyond our analysis, \citet[][p. 3]{blackwelletal2024} show that HTEs are the modal way that authors in these journals analyze causal mechanisms.

\begin{table}
\resizebox{\textwidth}{!}{
\begin{tabular}{l|ccc|cc|c} \hline
&\multicolumn{3}{c|}{Number of articles:} & \multicolumn{2}{c|}{Prevalence of HTEs} & Importance \\ \cline{2-7}
& & Quantitative & Reporting&Pr(Report HTEs $\mid$& Pr(Mechanism test $\mid$  & Pr(In abstract $\mid$\\
Journal (Volume) & Total & empirical & HTEs & Quant. empirical) & Report HTE) & Mech. test)  \\ \hline

\emph{AJPS} (65) & 61& 41 & 24 & 0.59 &  0.83 & 0.65\\
\emph{APSR} (115) & 102 & 75 & 42 & 0.56 & 0.90  & 0.63\\
\emph{JoP} (83) & 142 & 106 & 59 & 0.56 &0.76  & 0.82\\\hline
Total & 305 & 222 & 125 & 0.56 & 0.82 &  0.72\\ \hline
\end{tabular}}
\caption{Authors' classification of articles published in three leading political science journals in 2021.}\label{tab:class}
\end{table}

Clearly, the use of HTEs for mechanism attribution is quite common. However, it may be the case that these tests are viewed as secondary in importance to the ``main effect'' or the treatment effect in the full sample. While there is likely variation in the extent to which readers value main effects versus HTEs, authors commonly emphasize these mechanism tests as important results. The right column of Table \ref{tab:class} shows that conditional on relying on HTEs for mechanism evaluation, a majority (72\%) of articles mention these results about mechanisms in the abstract. Given word constraints, we view this as evidence that authors place non-trivial weight on HTE-based tests of mechanisms as important results of their analyses. 

We ask an essential but as-yet-unanswered question: under what conditions do HTEs provide evidence of mechanism activation? We do so by extending the workhorse causal mediation framework \citep{imai2010general}. We define a \emph{mechanism} as an underlying process that influences experience in order to produce a (causal) effect when activated \citep{sloughtyson2023}. In order to use HTEs to detect mechanisms, empiricists rely on a measured \emph{moderator}, or pre-treatment variable, that is thought to predict the degree to which treatment activates a mechanism and/or the degree to which the mechanism affects on the outcome of interest. Our framework provides a minimal structure necessary to link moderators to mechanisms to understand what can be learned from the presence or absence of HTEs. Our results characterize the conditions under which HTEs---specifically, a difference in conditional average treatment effects (CATEs) at different levels of a moderator---are sufficient to show that a specific mechanism is active. %We consider a mechanism to be active when its indirect effect is non-zero for some unit.

Our framework first elucidates the assumptions that are invoked for mechanism attribution. Specifically, if a covariate moderates the effect of a mechanism, it cannot also be a moderator for any other mechanism(s). If it were, there would be no way to determine which mechanism is responsible for the observed heterogeneity. Our first identification result shows that a difference in CATEs is generically equivalent to the difference in the conditional average indirect effect attributable to a mechanism if and only if exclusion assumptions hold. This means that the covariate does not moderate the effects of other mechanisms. Comparing these assumptions to those invoked by other methods for the quantitative study of mechanisms, namely the assumption of sequential ignorabililty in causal mediation analysis, the exclusion assumptions  are neither (logically) stronger nor weaker than sequential ignorability. This means that using HTEs for mechanism detection is not more or less agnostic than mediation analysis. However, one crucial benefit of HTEs is that they can provide information about mechanisms without measurement of mediators \citep[][p. 785]{imai2011unpacking}. 

Learning about a mechanism from HTEs requires more than exclusion assumptions. We introduce the concept of a mechanism detector variable (MDV), a covariate that predicts a stronger activation of a mechanism or a stronger effect of a mechanism on an outcome. Under the exclusion assumptions, the existence of HTEs at different levels of a covariate reveals that the covariate is an MDV. This is, in turn, sufficient to infer that the mechanism in question is active for at least one unit in the sample. This accords with current interpretation of mechanism tests that rely on HTEs. However, if HTEs do not exist for a covariate that is a candidate MDV, we do not learn whether the mechanism is active or inert. The lack of heterogeneity could be a consequence of misspecification of the theory, a mechanism that produces the same effect for all units, or an inert/inactive mechanism. This limits our ability to rule out the activation of a mechanism using HTEs.

Finally, we show how choices about which outcomes to measure and how to measure them can limit the use of HTEs as a test of mechanisms. Specifically, the exclusion assumptions that we articulate imply that the effects of a mechanism on an outcome should be additively separable from the effects of other mechanisms on that outcome. But common non-linear transformations of the outcome variable---e.g., binning, logging, or winsorization---violate this property of additive separability, generating heterogeneity that is not informative about mechanism activation.\footnote{\label{fn:functionalform}Our focus on additive separability stems from the widespread current practice of comparing conditional treatment effects across subgroups. Different measures of mechanistic influence will generally imply different functional form assumptions.} This finding suggests that the way that we measure the effects of a mechanism can limit our ability to attribute those very effects to the mechanism. It reveals the need to be more explicit about the relationship between mechanisms and measured outcomes than is common practice.

This paper makes three principal contributions. First, we introduce a new framework to understand the theoretical relationship between causal mechanisms and treatment effect heterogeneity. A large methodological literature provides guidance on the estimation of heterogeneous treatment effects (or ``interaction effects'') \citep[e.g.,][]{bramboretal2017,berryetal2009,hainmuelleretal20,grimmeretal2017,atheyetal2019}. More recent contributions the use of HTEs for extrapolation, prediction, or targeting of treatments \citep[e.g.,][]{egamihartman2020,huang2024sensitivity,devauxegami2022,kitagawatetenov2018,atheywager2021}. Yet, because these contributions are primarily statistical, they do not facilitate theoretical analysis of the relationship between mechanisms and the (measured) treatment effects that they produce. Our primary results on the theoretical properties of HTEs---and the problems with existing practice---therefore complement our understanding important statistical issues associated with the estimation of HTEs, namely limited statistical power \citep[e.g.,][]{mcclellandjudd376} and multiple-comparisons problems \citep{gerbergreen2012,leeshaikh2014,finketal2014}. 

Second, we expand a growing literature on the theoretical implications of empirical models (TIEM) \citep{bdmtyson2020, ashworth2021theory,ashworthetal2023,abramson2022we,slough2022phantom}. We make two central interventions to this literature. First, our framework makes explicit links between a causal mediation framework that is used more prominently by empiricists and (formal) theoretical models. This analysis complements recent work by \citet{blackwelletal2024}, who make explicit the strong assumptions that underpin the use of treatment effects on intermediate outcomes to discern causal mechanisms. Second, we introduce questions about how measured outcomes relate to theoretical constructs. While measurement is central to recent TIEM work on evidence accumulation in a cross-study environment \citep{sloughtyson2023ajps, sloughtyson2023, slough2025sign}, it has not been widely explored in the single-study environment. 

Finally, we provide practical guidance for empirical researchers who want to learn about which mechanisms generate observed effects. We illustrate this guidance concretely by analyzing four recent empirical studies about how partisan affinity (or bias) condition voter responses to corruption revelation \citep{anduizaetal2013,ariasetal2019,defigueredoetal2023,eggers2014}.
Our assumptions and results reveal a minimal set of attributes of an applied theory that can support the use of HTEs to learn about mechanisms. Two of these attributes, the relationships between (1) a covariate of interest and other mechanisms and (2) measured outcomes and theoretical objects of interest, are generally not discussed in applied work. Second, we show how interpretation of HTEs can be improved, returning to the statistical problems that are well known in this literature. Third, we discuss how our analysis can be used to inform prospective research design. Finally, we consider the merits of invoking additional assumptions or statistical models to address some of the issues we identify. Collectively, these suggestions allow practitioners to accurately use---or, when indicated, avoid---HTEs as a quantitative test of mechanisms.

\section{Motivating Example: Corruption Revelation and Voting}

Despite widespread unpopularity of corruption by public officials in public opinion polls, voters in many democracies routinely re-elect politicians who engage in corruption. One explanation for the prevalence of public corruption in many democracies is that corruption by specific politicians is not observed by voters. If voters were to receive information that an incumbent or candidate were corrupt, then they should exhibit less support for the candidate. Scholars have evaulated the empirical merits of this argument using field experiments \citep[e.g.,][]{dunningetal2019}, randomized corruption audits \citep[e.g.,][]{ferrazfinan2008}, survey experiments \citep[e.g.,][]{anduizaetal2013}, and observational studies \citep[e.g.,][]{eggers2014}. Existing meta-studies document heterogeneity in voter responses to such information \citep{incerti2019,slough2024}. Here, we examine whether such heterogeneity is informative about a voter (dis)taste for corruption mechanism.

Specifically, consider a randomized experiment in which the treatment reveals that an incumbent in a constituency has engaged in corruption. We will say that $c_i=1$ if voter $i$ receives the corruption information treatment and $c_i=0$ if they do not receive the information. Voters value multiple attributes of politicians. First, they dislike corruption, albeit to different degrees, where $\lambda_i\in(0,1)$ captures variation in corruption aversion across voters. %\footnote{The upper bound on $\lambda_i$ expresses the idea that voters care about more than corruption.} 
Second, voters value partisan alignment with politicians. Specifically, voter $i$  may be aligned $(a_i=1)$, independent/neutral ($a_i=0$), or unaligned $(a_i=-1)$ with an incumbent politician. Finally, voters idiosyncratically value valence characteristics of a politician (e.g., their personality or personal background), which we represent with the random variable $\varepsilon_i$. We assume that $\varepsilon_i$ is independent of $\lambda_i$ and $a_i$, and follows a standard normal distribution. Each voter's utility from a vote for the incumbent is given by:
\begin{align*}
u_i &= -\lambda_i c_i + a_i + \varepsilon_i
\end{align*}

Corruption information enters voters' assessment of the incumbent through their distaste for corruption, $\lambda_i$. One can think of the product $\lambda_i c_i$ as akin to a mediator that captures the degree to which this distaste is activated by the information treatment $c_i$. Ultimately voters choose between the incumbent and a challenger. Without loss of generality, we normalize a voter's utility from a vote for the challenger to be 0. Therefore, voters will vote for the incumbent if $u_i \ge 0$.

\subsection{From Theory to Empirical Research Design}

Mapping this model onto the empirical research design, we consider two measures of voting outcomes. The first outcome, $y_{i1}$, measures voters' (expected) utility from the incumbent. It is obviously difficult and rare to measure utility from a candidate directly. \label{comment:likert}However, one could, in principle, ask a voter to evaluate their incumbent on a 0-100 scale (or elicit willingness-to-pay for the incumbent's reelection). As in the model, the voter's expected utility from a vote for the incumbent is given by:
\begin{align}\label{eq:utility}
y_{i1}(c) &= -\lambda_ic_i+a_i+\varepsilon_i
\end{align}
The second outcome, $y_{i2}$, measures each voter's (self-reported) vote choice for the incumbent. Vote choice is a commonly measured outcome in literature on voter behavior. This outcome is given by:
\begin{align}
y_{i2}(c) &= \begin{cases}
1 &\text{ if } -\lambda_ic_i+a_i+\varepsilon_i \geq 0\\
0 &\text{ else }
\end{cases}
\end{align}

The data generating process of  the model is depicted in Figure \ref{si:fig_dag}. As we report in Table \ref{tab:class}, empirical researchers often turn to estimation of heterogeneous treatment effects with respect to a pre-treatment covariate to assess the activation of a mechanism. In the context of an experiment, heterogeneity is assessed through the estimation of conditional average treatment effects (CATEs) at different levels of a pre-treatment covariate, $X_k$. Given outcomes $j\in \{1,2\}$:
\begin{align}
CATE(y_j, X_k)= E[y_{ij}(c_i = 1)-y_{ij}(c_i = 0)|X_{ik}= x]
\end{align}
We will say that treatment effects are heterogeneous if for some $x, x' \in X_k$ where $x\neq x'$, $CATE(y_j, X_k = x)- CATE(y_j, X_k = x') \neq 0$. Importantly, whether treatment effects are \emph{homogeneous} or \emph{heterogeneous} in a given covariate, is fundamentally a qualitative classification. This means that even though researchers are using a quantitative metric to evaluate mechanisms, they are fundamentally making a \emph{qualitative} inference about mechanism activation.

Recall that standard interpretations in the empirical literature use the presence of HTEs as evidence that a mechanism is active and the absence of HTEs to assert that a mechanism is inert. Given our model of voter behavior, we ask: do heterogeneous treatment effects provide evidence that the relevant mechanism is indeed voter distaste for corruption? To develop intuitions, we evaluate four HTE combinations using moderators  $X_1=\Lambda$ and $X_2= A$, where $\Lambda$ is the set of all possible values of $\lambda_i$ (corruption aversion) and $A$ is the set of all possible values of $a_i$ (partisan alignment with the incumbent), and outcomes $y \in \{y_1, y_2\}$. Remark \ref{remark1} shows that for the outcome measuring a voter's expected utility from a vote for the incumbent ($y_1)$, heterogeneity in CATEs correctly provides evidence that the mechanism is voter distaste for corruption, not some type of corruption-induced partisan-realignment or biased learning that depends on partisanship \citep[e.g., motivated reasoning in ][]{littleetal2022} (for example) that are not present in our model\label{r2:motreasoning}.

\begin{remark}\label{remark1}For the outcome measuring voter preferences, $y_1$,
    
    (a)  Given $\lambda > \lambda' \in \Lambda$,  $|CATE(y_1, X_1 = \lambda)|>|CATE(y_1, X_1 = \lambda')| $.
    
    (b)  Given $a \neq a' \in A$,  $CATE(y_1, X_2 =a)-CATE(y_1, X_2 =a') = 0$.
    
    (c)   If $CATE(y_1, X_k=x )-CATE(y_1, X_k=x') \neq 0$, then $X_k=X_1=\Lambda$. % remove the possibility that m \in Pi and m' \in V; this is not interesting. 

\noindent(All proofs in appendix.)
\end{remark}

Researchers will detect heterogeneity with the corruption aversion moderator $\lambda$ (by (a)). Voters whose corruption aversion  $\lambda$ is larger value the incumbent less (upon receiving the information) than voters with lower corruption aversion. HTEs are not observed for the partisan alignment (non)-moderator $a$ (by (b)) because ideological alignment (not the mechanism) and distaste for corruption (the mechanism) are additively separable in voters' utility (in \eqref{eq:utility}). Here, (c) shows that researchers are unlikely to mis-attribute the mechanism through conventional interpretation of HTEs. 

However, for the vote choice for the incumbent, $y_2$, the results from Remark \ref{remark1} change. First, researchers may observe HTEs for different levels of partisan alignment, $a$, as well as for different levels of corruption aversion, $\lambda$. A na\"{i}ve interpretation might suggest that the effect of the corruption information \emph{does} work through some channel involving ideological re-alignment or biased learning in addition to a channel involving distaste for corruption. 

\begin{remark}\label{remark2}For the outcome measuring voter choice $y_2$,
    
    (a)  For $\lambda > \lambda' \in \Lambda$,  $|CATE(y_2,  \lambda)|>|CATE(y_2,  \lambda')|$.
    
    (b)  For $a \in A$,  $|CATE(y_2, a=-1)|<|CATE(y_2, a=1)|<|CATE(y_2, a=0)|$.
    
    (c)   If $CATE(y_2, X_k=x)-CATE(y_2, X_k=x' ) \neq 0$, then $X_k = X_1=\Lambda$ or $X_k =X_2=A$.
    
\end{remark}

Voters with stronger aversion to corruption exhibit larger treatment effects on their voting behavior (by (a)). Interestingly, even though distaste for corruption is the unique mechanism (in the model) through which information affects vote choice, we now observe heterogeneous treatment effects across voters of different partisan alignments. Specifically, the treatment effect is largest among neutral/independent voters ($a=0$) voters, while strongly aligned partisans ($a=-1,1$) exhibit smaller changes in voting behavior, by (b). Moreover, the effect is greater for incumbent-aligned voters than for challenger-aligned voters. This pattern is driven by a measurement concern: ceiling and floor effects that emerge when utility is transformed into vote choice. Aligned voters are substantially more likely to vote for or against the incumbent regardless of new information, while moderates hover around the decision threshold, making them more responsive to changes in observed corruption.

Importantly, the smaller swings observed among strong partisans are not due to the presence of an additional mechanism. Rather, it is because their prior disposition already places them near the upper or lower bounds of the probability scale in terms of propensity to vote for the incumbent. Researchers who are unaware of the underlying mechanism might mistakenly conclude that partisan voters are ``less responsive" to treatment and, as a result, wrongly infer that they are engaged in biased learing/insensitive to information or that they are less prone to partisan realignment on the basis of corruption information.

More concretely, we see HTEs in partisan affiliation ($a$) for vote choice because voters make a binary choice between the incumbent and challenger. But this binary choice means that the distaste for corruption mechanism and the partisan alignment predictor (which does not moderate the mechanism) are no longer additively separable with respect to vote choice. Consequently, researchers are apt to detect HTEs in partisan alignment even when it does not moderate any mechanism. \label{e:naive}A na\"{i}ve interpretation of this test would lead to a Type-I error in our inference about the activation of a mechanism involving partisan alignment.

\label{comment:sim}While Remarks \ref{remark1} and \ref{remark2} rely on theoretical analysis of a model, one might ask whether an empirical researcher would be likely to detect this heterogeneity in their data. We therefore conduct a Monte Carlo simulation in Appendix \ref{app:sim} that uses our theoretical model to guide the data generation process from a hypothetical experiment. We vary: (1) (average) voter support for the incumbent in the electorate ($\overline{y_2 (0)} \in \{0.1,0.2,...,0.9\}$) by varying the mean of the valence distribution; and (2) sample size in the hypothetical experiment ($n \in \{100, 500, 1,000\})$. Consistent with our theoretical results, we observe  HTEs in the corruption aversion covariate for both outcomes (voter utility and vote choice) and HTEs in partisan affiliation only for the vote choice outcome. More importantly, with the finite sample sizes that we simulate, we show that for high and low values of voter support for the incumbent in the electorate, we are \emph{more likely (better powered) to detect HTEs in partisan affiliation} than in corruption aversion for the vote choice outcome. This finding holds across multiple estimators of the difference in CATEs (treating $a_i$ linearly or as a factor variable). Thus, we cannot solely rely on the statistical properties of our designs to help us screen heterogeneity that is informative about mechanism activation from heterogeneity that is not.

This example yields three important observations that we develop by proposing a new framework: 
\begin{enumerate}
\item The use of HTEs does, in some cases (i.e., Remark \ref{remark1}), provide information about mechanism activation. This accords with current practice.
\item The use of HTEs to measure mechanism activation relies on assumptions about the relationship between moderators and mechanisms of interest which are typically implicit. 
\item The contrast between Remarks \ref{remark1} and \ref{remark2} in which the theory (and thus mechanism) is fixed but outcomes differ shows that the use of HTEs for assessing mechanism activation depends on the measurement of outcomes of interest. 
\end{enumerate}

\section{Framework}

\subsection{Defining HTEs}

Our framework is built upon the potential outcomes framework or Neyman-Rubin causal model \citep{neyman1923applications,rubin1974estimating}. \label{comment:z}We denote a randomly assigned treatment by $Z \in \{z, z'\}$.\footnote{All results hold for observational studies in which treatment is conditionally independent of potential outcomes, given a conditioning set of covariates, $X$.} In order to consider HTEs with respect to pre-treatment moderators, denote the vector of pre-treatment covariates by $X=(X_1,X_2,...,X_K) \in \mathbb{R}^K$. Some of $X_k$ are \emph{measured}. To improve readability and emphasize the key elements of the framework, we omit subscript $i$ from the notation throughout the framewrok and results sections.

A valid mediator, or mechanism representation, should: (1) be affected by treatment, $Z$, and (2) have a non-zero effect on the outcome. Further, the effect of treatment on a mediator or the mediator's effect on the outcome could vary with some covariate(s), $X_k$. We define a mediator as a function denoted by $M(Z; X)$.\footnote{Note that potential outcomes are often written as a function of only the manipulated treatment, e.g., $M(Z)$ and $Y(Z)$, to reflect ``no causation without manipulation'' \citep[][p. 959]{holland1986}. We choose to denote covariates $X$ as arguments to both potential outcomes in order to clarify the structure of covariates and mediators in our potential outcomes.} It represents the potential outcomes of causal mediator given treatment $Z$ and covariates $X$. While we employ mediators to clarify the link between HTEs and mediation, these mediators (or mechanisms) do not need to be measured in order to analyze HTEs. Finally, we denote a potential outcome by $Y(Z, M; X)$. 

We consider the practice of quantifying heterogeneous treatment effects by estimating and comparing CATEs, as documented in Table \ref{tab:class}. Following this convention, we consider HTEs with respect to pre-treatment moderators, i.e., for some variable $X_k$, where $k \in \{1,2,...,K\}$. %\footnote{It is easy to generalize to a subset of covariates $X_k$, where $k \subseteq \{1,2,...,K\}$.}

\begin{definition}[Conditional Average Treatment Effect] \label{def:cate}Consider pre-treatment covariate $X_k$. Given that $z \neq z' \in Z$, the conditional average treatment effect (CATE) of $Z$ on $Y$ when $X_k = x$ is:
\[CATE^Y(X_k = x) = E_{X_{\neg k}} [Y(Z = z, M; X)-Y(Z = z',M; X)| X_k = x]\] 
\end{definition}
Note that the CATE of a treatment $Z$ is defined with respect to (potential) outcome variable $Y$. We will index estimands by the outcome variable of interest throughout. Given this definition, there exist HTEs when CATEs of $Z$ on $Y$ differ at different values of a covariate $X_k$:

%\begin{definition}[Heterogeneous treatment effects]\label{def:hte} HTEs exist with respect to pre-treatment covariate $X_k$ if $CATE_Y(X_k = x) \neq CATE_Y(X_k = x')$ for almost every pair $x\neq x' \in X_k$. 
%\end{definition}

\begin{definition}[Heterogeneous treatment effects]\label{def:hte} HTEs on $Y$ exist with respect to pre-treatment covariate $X_k$ if $CATE^Y(X_k = x) \neq CATE^Y(X_k = x')$ for some $x\neq x' \in X_k$.\footnote{More precisely, the probability measure for the set that contains such $x$ and $x'$ is non-zero.} 
\end{definition}

Definitions \ref{def:cate}-\ref{def:hte} formalize the current practice of comparing CATEs to evaluate whether treatment effects exhibit heterogeneity. However, while our potential outcomes implicitly depend on a mediator, $M$, there is not yet a link to an underlying mechanism.

\subsection{Causal Mediation and Indirect Effects}\label{sec:mediation}

We now develop the mapping between analysis of treatment effect heterogeneity and causal mediation, which seeks to quantify the effect of one or more mechanisms. Specifically, consider a decomposition of the total effect (on a unit, $i$) into direct and indirect effects \citep{imai2013identification}.  Suppose that there exist two mediators (or mechanisms), indexed by $M_1, M_2$. Given two treatment values, $z, z' \in Z$, the total effect of $Z$ on $Y$ is:
\begin{align}
TE^Y(z, z'; X) &= Y(z, M_1(z; X), M_2(z; X); X)- Y(z', M_1(z'; X), M_2(z'; X); X)
\end{align}

Our notation varies slightly from conventional presentations of mediation that only consider one mechanism (mediator) \citep{imai2010general}. In the main text, we describe the case with two mechanisms since it is straightforward to generalize this to the special case of one mechanism or to a setting with more than two mechanisms. Further, as above, we continue to index causal effects by the outcome variable, here $Y$. Treatment effects on $Y$ may consist of direct ($DE^Y$) and indirect ($IE^Y_1$ and $IE^Y_2$) effects, as follows: \footnote{See \citet{acharya2016explaining} for the identification of the controlled direct effect.}
\begin{align}
DE^Y(z, z'; X) = Y(z, M_1(z; X),  M_2(z; X); X)- Y(z', M_1(z; X), M_2(z; X); X)\\
IE^Y_1(z, z'; X) = Y(z', M_1(z; X), M_2(z; X) ; X)- Y(z', M_1(z'; X),M_2(z; X) ; X)\\
IE^Y_2(z, z'; X) = Y(z', M_1(z'; X) ,M_2(z; X) ; X)- Y(z', M_1(z'; X), M_2(z'; X) ; X)
\end{align}

The direct effect, $DE^Y(z, z';X)$ represents the direct effect of $Z$ on $Y$ holding both mediators $M$ at potential outcomes $M(z; X)$. Our preferred interpretation of the direct effect is the composite effect of any other mechanisms (aside from $M_1$ and $M_2$). However, the direct effect could also include unmediated effects of treatment on an outcome. In our motivating example, there is only one mechanism---voter distaste for observed corruption---so the direct effect (all other mechanisms) is zero. This is evident because the treatment has zero effect if were fix the mediator, $\lambda_ic_i$, to a given level. The indirect effect of mechanism $j \in \{1,2\}$ measures the effect on the outcome that operates by changing the potential outcome of mediator $M_j$. In our example, the indirect effect measures the effect that passes through the distaste mechanism. 

As is standard, we can re-write the total effect as follows:\footnote{Here, we assume that the direct and indirect effects do not vary at different levels of Z. See more discussion by \citet{imai2011unpacking}. All results in the paper hold with other decompositions, albeit with slightly different interpretations. See additional discussions of general cases in the \ref{app:multi} and correlated mechanisms in \ref{app:mediation}.}
\begin{align}\label{eq:te}
TE^Y(z, z';X) &= DE^Y(z, z';X) +  IE^Y_1(z,z'; X) + IE^Y_2(z,z'; X), 
\end{align}
which is defined at the unit, or individual level. If we evaluate expectations over $X$, we obtain:
\begin{align}
ATE^Y(z, z') &= E_{X}[Y(z,M_1(z;X),M_2(z;X);X)- Y(z',M_1(z';X),M_2(z';X);X)] \\
&= E_{X}[DE^Y(z, z'; X) +  IE^Y_1(z,z'; X)+IE^Y_2(z,z'; X)] \label{equ:ade}\\
&:=  ADE^Y(z, z')+ AIE^Y_1(z,z')+AIE^Y_2(z,z')\label{equ:aie}
\end{align}

We use $ADE^Y$ and $AIE^Y_j$ to denote average direct effect and average indirect effect of mechanism $j$ for outcome $Y$, respectively. Throughout the paper, we assume the expectation in \eqref{equ:ade} is well-defined. 

\subsection{HTEs and the Identification of Indirect Effects}\label{ss:identification}

How do HTEs---differences in CATEs---relate to the indirect effects that are estimated within the mediatiation framework? To show this relationship, consider a decomposition of a CATE into a conditional ADE and two conditional AIEs following \eqref{equ:aie}\footnote{Conditional $ADE^Y(z,z';X_k=x)$ and $AIE_j^Y(z,z';X_k = x)$ evaluate expectations over all $X$ but fixing $X_k$ at $x$.}:
\begin{align*}
CATE^Y(X_k=x)= ADE^Y(z,z';X_k=x) + \sum_{j=1}^2 AIE^Y_j(z,z';X_k = x)
\end{align*}
One can then express the difference in CATEs, at two distinct levels of $X_k$ as:
\begin{align}
\begin{aligned}
CATE^Y(X_k=x)-CATE^Y(X_k = x')=&[ ADE^Y(z,z';X_k=x) - ADE^Y(z,z';X_k=x')]  \\ &  +   \sum_{j=1}^2 [AIE^Y_j(z,z';X_k = x)- AIE^Y_j(z,z';X_k = x')]   \label{eq:catedif}
\end{aligned}
\end{align}

From \eqref{eq:catedif}, it is clear that HTEs could arise from differences in direct and/or indirect effects. Suppose that we were interested in evaluting the activation of mechanism 1 with respect to outcome $Y$. Expression \eqref{eq:catedif} shows that we cannot automatically attribute observed heterogeneity to mechanism 1. Instead, to link a difference in CATEs (HTEs) to a difference in indirect effects, we need to assume that the conditional direct effect $ADE^Y(z,z';X_k = x)$ and any conditional indirect effect(s) of the other mechanism(s) $AIE^Y_2(z, z';X_k = x)$ do not vary at different levels of $X_k$, as stated in Assumptions \ref{cond:irre1}-\ref{cond:irre2}. Assumption \ref{cond:irre1} formalizes an exclusion assumption posited by \citet{imai2011unpacking}, who focus on the case of a single mechanism.\footnote{Specifically, \citet[][p. 784]{imai2010general} write ``if the size of the ADE does not depend on the pretreatment covariate ... a statistically significant interaction term implies that the [average causal mediation effect] is larger for one group ... than for another group.''\label{fn:imai}} 

\begin{ass}[Exclusion I]\label{cond:irre1}
Given $z,z' \in Z$ and $x,x' \in X_k$, $X_k$ is excluded from the direct effect such that $ADE^Y(z, z'; X_k = x)=ADE^Y(z, z'; X_k = x')$.
\end{ass} 

\begin{ass}[Exclusion II]\label{cond:irre2}
	Given $z,z' \in Z$ and $x,x' \in X_k$, $X_k$ is excluded from the indirect effect of the other mechanism 2: $AIE^Y_{2}(z, z'; X_k=x)=AIE^Y_{2}(z, z'; X_k=x')$.
\end{ass}

Assumptions \ref{cond:irre1}-\ref{cond:irre2} constrain the relationship between a moderator, $X_k$, other mechanisms (e.g., $M_2$) and any direct effect of treatment.\footnote{\label{fn:implicitmed}Other methods for mechanism detection invoke distinct exclusion assumptions. For example, the scaling stage of implicit mediation uses instrumental variables analysis to estimate causal mediation effects \citep{bullockgreen2021}. In contrast to our analysis of HTEs, the exclusion assumptions underpinning implicit mediation restrict the mediators that might be activated by a treatment. See also \citet{fu2024extracting} for a related discussion on mediation analysis with HTEs.} Figure \ref{fig:assdag} illustrates these assumptions graphically. \label{comment:dag}While this figure resembles a directed acyclic graph (DAG), we depart from conventional presentation of DAGs as vertices (nodes) and edges (arrows) because there are edges that point to other edges (rather than vertices). We make this departure because our assumptions impose greater structure on the possible causal moderation (or lack thereof) than is assumed in traditional DAGs. This departure is not new. \citet{nilssonetal2021} notes that there is no standardized representation of causal moderation in DAGs, so our graphs are informed by an existing proposal for the representation of these effects by \citet{weinberg2007}. This notation allows us to accurately convey the structure of causal moderation.

Once Assumptions \ref{cond:irre1} and \ref{cond:irre2} are invoked, it is straightforward to see that difference in CATEs reduces to differences in $AIE^Y_1$ at different levels of $X_k$, which we state formally in Proposition \ref{prop:ident}. Assumptions \ref{cond:irre1}-\ref{cond:irre2} are generically necessary because it is possible that $ADE^Y(x)-ADE^Y(x') \neq 0$ and $AIE^Y_2(x)-AIE^Y_2(x')  \neq 0$ exactly offset each other. But under generic parameter values, the probability of this knife-edge event is zero.

\begin{figure}
\begin{center}
\begin{tikzpicture}
\node at (0,0) (z) {$Z$};
\node at (3,1.1) (m) {$M_1$};
\node at (3,0) (m2) {$M_2$};
\node at (6, 0) (y) {$Y$};
\node at (1,2.5) (x1) {$X_k$};
\draw[->] 
    (z) edge [bend left = 22.5] (m)
    (m) edge [bend left = 22.5] (y)
    (z) edge [bend right = 35] (y);
\draw[->, dashed, color = blue]
    (x1) edge [bend right = 20] (1, -.5);
\draw[->, dash dot, color = red]
    (x1) edge [bend left = 20] (2, 0.1);
\draw[->, dash dot, color = red]
    (x1) edge [bend left = 20] (4, 0.1);

\draw[->] 
    (x1) edge [bend left = 30] (y);
\draw[->] 
    (x1) edge [bend left = 10] (m2);
\draw[->] (x1)--(1.25, .95);
\draw[->] (z)--(m2);
\draw[->] (m2)--(y);
\end{tikzpicture}
\end{center}
\caption{Assumption \ref{cond:irre1} rules out the blue dashed path. Assumption \ref{cond:irre2} rules out both of the red dot-dashed paths. All black solid paths are permissible under Assumptions \ref{cond:irre1} and \ref{cond:irre2}.} \label{fig:assdag}
\end{figure}

%\begin{prop} \label{prop:ident}
%	Generically, $CATE(X_k=x)-CATE(X_k = x')= AIE_1(z,z';X_k = x) - AIE_1(z,z';X_k = x')$ if and only if assumptions \ref{cond:irre1}-\ref{cond:irre2} hold.
%\end{prop}

\begin{prop} \label{prop:ident}
	Assumptions \ref{cond:irre1}-\ref{cond:irre2} are sufficient and generically necessary for $CATE^Y(X_k=x)-CATE^Y(X_k = x')= AIE^Y_1(z,z';X_k = x) - AIE^Y_1(z,z';X_k = x')$.
\end{prop}

Proposition \ref{prop:ident} clarifies that a difference in CATEs does not identify either conditional AIE of mechanism 1 ($AIE^Y_1(z, z'; X_k = x)$ or $AIE^Y_1(z, z'; X_k = x')$) in the absence of further assumptions. Rather, the difference in CATEs identifies a \emph{difference} in conditional AIE's. Thus, identification of this difference is not sufficient to identify indirect effects, as is the goal in (standard) mediation analysis. However, it is straightforward to see that if the difference in AIEs is not equal to zero, there must exist some unit for whom the AIE is not equal to zero. A non-zero difference in AIEs is therefore a sufficient condition for the activation of the relevant mechanism for at least one unit. This identification result motivates a more precise version of our research question: ``Under what conditions are HTEs with respect to a covariate $X_k$ sufficient to show that there exists some unit for which $IE^Y_1(z,z'; X_k) \neq 0$?'' 

To understand our later results, it is useful to see how Assumptions \ref{cond:irre1} and \ref{cond:irre2} could be violated. The first and most obvious violation would be that a covariate $X_k$ moderates multiple mechanisms (or one mechanism and the direct effect). This is clear from Figure \ref{fig:assdag}. A second and less obvious violation is that the additive separability of the (indirect) effects of mechanisms breaks down. This would mean that if some $X_k$ moderates the effect of mechanism $1$, then it must also moderate the effect of any other mechanism that ``interacts with'' or whose (indirect) effect depends on the effect of mechanism $1$. In Figure \ref{fig:assdag}, this would occur if, for example, there was an interaction between the effects of $M_1$ and $M_2$ on the outcome $Y$. %Importantly, this implies that for the same structure of mechanisms, (i.e., the same treatment, covariate(s), and mediator(s)), Assumptions \ref{cond:irre1}-\ref{cond:irre2} may hold for some outcome variable(s) but not others.\footnote{Note that there is no guarantee there exists an outcome for which Assumptions \ref{cond:irre1}-\ref{cond:irre2} hold in a given system.} Thus, understanding the link between a mechanism and its measured outcomes is essential to understand to evaluate whether a mechanism is active.\footnote{This problem is underrecognized in other approaches to the study of mechanisms including mediation analysis. We return to this point at the conclusion of this paper.}

It is important to note that mediation analysis does not invoke Assumption \ref{cond:irre1} or \ref{cond:irre2}, and instead invokes an assumption of sequential ignorability \citep{imai2010identification}. There is no logical ordering of the two types of assumptions: the exclusion assumptions do not imply sequential ignorability, nor does sequential ignorability imply the exclusion assumptions. This means that HTEs cannot said to be a more or less agnostic test of mechanism activation than mediation. In some applications, one set of assumptions may be more plausible or defensible than the other, but we cannot make a general claim about the strength of these distinct sets of assumptions. We provide a broader discussion comparing the use of HTEs to mediation analysis in \ref{app:mediation}.

\subsection{Connecting Mechanisms to Measured Variables}

\label{comment:rv}While our identification results show that HTEs \emph{can} provide information about mechanism activation for some unit(s), the framework does not yet elucidate the relationship between a mechanism and measured variables. To understand why it is critical to develop this relationship beyond our identification results, note that a randomly-generated covariate that is independent of all variables in a research design (or system) satisfies Assumptions \ref{cond:irre1} and \ref{cond:irre2} by construction. Here, we would not expect to observe HTEs at different levels of the randomly-generated variable. But this lack of heterogeneity should not be informative about the substantive mechanism(s) at play.

\label{comment:metaphysical}We view a mechanism as an underlying process that responds to some activation and produces a given set of outputs. In the context of our running example, the voter distaste mechanism is activated by the observation of corruption information about the incumbent. It produces a number of outputs---a voter's assessment of the incumbent and their voting decision---among other (unmodeled) possibilities. None of these objects---a voter's information, their utility, or their voting decision---needs be inherently quantitative (though they could be). As social scientists, we choose how to measure and operationalize each of these objects: we normalize the voter information treatment to a binary (0/1) scale for convenient estimation of treatment effects; we choose some type of Likert scale to measure voter assessments/utility; and we choose self-reported vote choice or aggregated voting results to measure vote decisions. These operationalizations facilitate our ability to measure or quantify the effect of a mechanism on an outcome using, for example, a difference in CATEs. This implies that our ability to observe a mechanism's effect depends fundamentally on \emph{how} we choose to measure it \citep{sloughtyson2023}. Our concern here is therefore the link between a substantive mechanism and the measured variables in our framework. 

\subsubsection{Outcome variables and mechanisms}

First, consider the relationship between measured outcomes and the outputs of a mechanism. A given mechanism produces multiple possible outputs; a researcher chooses to operationalize and measure some subset of those outputs as outcome variables. But not all outcome variables relate to the underlying mechanism in the same way. For example, in our motivating example the mechanism---voter distaste for observed corruption---affects a voter's utility from the incumbent ($y_1$). The second outcome, vote choice, is a deterministic but non-linear function of utility given by the function $y_2(c) = \mathbb{I}[y_1(c) \ge 0]$. We refer to $y_2(c)$ as a \textbf{transformed (potential) outcome} relative to $y_1$.

\begin{definition}\label{def:outcomes}
Given a (potential) outcome $Y(Z,M;X)$. Let $\widetilde{Y}(Z,M;X) = h(Y)$, where $h(\cdot)$ is a non-linear function. Then we call $\widetilde{Y}(Z,M;X)$ a \textbf{transformed (potential) outcome} relative to $Y(Z,M;X)$.
\end{definition}

\label{def:h}Formally, for a function $h(\cdot)$, we define the increment $\delta(t;d)=h(t+d)-h(t)$. We say that the function $h$ is non-linear if there exists nondegenerate $t, t'$, and $d \neq 0$, such that the increment is different, $\delta(t;d) \neq \delta(t';d)$.\footnote{We assume that $h(t+d)$ is well defined.} This is a general definition that does not assume differentiability of the function $h(\cdot)$. 

The non-linear transformation is important because it affects the validity of the exclusion assumptions. Specifically, suppose that we believed that the exclusion assumptions held for a given covariate, $X_k$, mechanism, $M_1$, and outcome $Y$.\footnote{Note that there is no guarantee there exists an outcome for which Assumptions \ref{cond:irre1}-\ref{cond:irre2} hold in a given system of mechanisms.} Given our definition of the non-linear function $h(\cdot)$, $ADE^{\widetilde{Y}}(z,z';X_k=x)=\mathbb{E}[\delta(Y(z',M_1,M_2;X_k = x);Y(z,M_1,M_2;X_k = x)-Y(z',M_1,M_2;X_k = x))]$. However, in general, $ADE^{\widetilde{Y}}(z,z';X_k=x) \neq ADE^{\widetilde{Y}}(z,z';X_k=x')$ without other specific assumptions about the functional form of $h(\cdot)$. This means that even if Assumption \ref{cond:irre1} holds for $Y$, it generally will not hold for the transformed outcome $\widetilde{Y}$. A similar result holds for Assumption \ref{cond:irre2} and $AIE^{\widetilde{Y}}_2$. The logic for these observations is that the non-linear transformation ``breaks'' the additive separability of the mechanism from (1) other mechanism(s) and (2) predictors of the outcome. This mechanically generates additional causal moderation with different mechanisms without changing the underlying processes through which treatment affects the outcome. Given our distinction between substantive mechanisms and measurement of the effects they produce, we view the introduction of heterogeneity via a change in outcome variables as distinct from the presence of heterogeneity that exists in how mechanisms present in the world. 

When are transformed outcomes used in applied work? Three cases are quite common (but non-exhaustive). The first is akin to our running example: a treatment induces a change in information or utility that then affects some discrete choice of strategy. The second holds that a treatment changes an actor's attitude ($Y)$. But since attitudes are latent, survey researchers measure changes in the attitude by employing a Likert scale of the form: 
\begin{align}\label{eq:likert}
    h(Y) &= \begin{cases}
        1 & Y \in (-\infty, c_1]\\
        2 & Y \in (c_1, c_2]\\
        \vdots & \\
        Q & Y \in (c_{Q-1}, \infty),
    \end{cases}
\end{align}
in which $c_t$ denotes increasing thresholds in a latent attitude. Third, a researcher may employ non-linear transformations of an outcome to demonstrate the robustness of results to measurement choices. Here, they may bin a count variable to capture the extensive margin of some behavior or winsorize or logarithmize a skewed outcome etc. In any of these cases, even if Assumptions \ref{cond:irre1}-\ref{cond:irre2} were to hold for the orignal outcome (utility, attitudes, or the raw variable), they will not hold for the transformed outcome.

\subsubsection{Moderators and mechanisms}

Second, consider the role of a measured covariate, $X_k$, which we seek to use to detect the activation of focal mechanism $M_1$. A mechanism detector variable for $M_1$ is a covariate that produces different conditional AIEs at different values. To economize notation, we will denote $AIE_j^Y(X_k =x)=AIE_j^Y(z,z'; X_k = x, X_{\neg k})$ as the average indirect effect of mechanism (mediator) when $X_k = x$. 

\begin{definition}\label{def:MDV}A pre-treatment covariate $X_k$ is a \textbf{mechanism detector variable} for mechanism $j$ with respect to outcome $Y$ if for some $x, x' \in X_k$, $AIE^Y_j(X_k = x) \neq AIE^Y_j(X_k = x')$.  %\footnote{We can also define MDV as $IE_j(X_k = x) \neq IE_j(X_k = x')$ for all individual $i$; but it is too strong because we mainly foucs on the (conditional) average treatment effects. Same for the Assumptions 1 and 2 below.}
\end{definition}

We then denote $\textbf{X}^{MDV} \subseteq \{X_1,X_2,...,X_L\}, L\leq K$, as the (possibly empty) set of covariates that satisfy Definition \ref{def:MDV} for the mechanism $j$. Intuitively, if $X_k \in \textbf{X}^{MDV}$, then covariate $X_k$ can serve as an indicator for a mechanism/mediator of interest.\footnote{A slightly stronger version of Definition \ref{def:MDV} holds when $Y$ is continuously differentiable with respect to $M$, $Z$, and $X_k$. In this case, Definition \ref{def:MDV} can be expressed as $\frac{\partial}{\partial X_k}\left(\frac{\partial Y}{\partial M}\frac{\partial M}{\partial Z}\right)\neq 0$ for some $x \in X_k$.} Under our definition of MDVs, it could be the case that $X_k$ moderates the effect of the treatment on the mediator. Interestingly, it could also be the case that $X_k$ moderates the effect of the mediator on the outcome. Both possibilities are depicted in Figure \ref{fig:MDVdag}.  Researchers using HTEs to investigate mechanisms seek to detect evidence of a mechanism using treatment-by-covariate interactions with a proposed MDV. In other words, given a proposed MDV $X_k$, we want to learn whether $X_k \in \textbf{X}^{MDV}$.

\begin{figure}
\begin{center}
\begin{tikzpicture}
\node at (0,0) (z) {$Z$};
\node at (2,0) (m) {$M$};
\node at (4, 0) (y) {$Y$};
\node at (1,1) (x1) {$X_{k}$};
\draw[->] (z)--(m);
\draw[->] (x1)--(1, 0.05);
\draw[->] (m)--(y);

\node at (7,0) (z1) {$Z$};
\node at (9,0) (m1) {$M$};
\node at (11, 0) (y1) {$Y$};
\node at (10,1) (x2) {$X_{k}$};
\draw[->] (z1)--(m1);
\draw[->] (x2)--(10, 0.05);
\draw[->] (m1)--(y1);
\end{tikzpicture}
\end{center}
\caption{The two Panels depict the causal structure of two MDVs for mechanism $M$ graphically. Both Panels are consistent with Definition \ref{def:MDV}. } \label{fig:MDVdag}
\end{figure}

\section{Results}
We consider the conditions under which HTEs (or lack thereof) are informative about the activation of a mechanism, $M_1$, for some unit(s) in a sample. To do so, we analyze four exhaustive and mutually exclusive cases that vary the (1) existance of HTEs and (2) whether an outcome is non-linearly transformed or not. In each case, we will assume that Assumptions \ref{cond:irre1}-\ref{cond:irre2} hold for the outcome $Y$. This follows the identification result in Proposition \ref{prop:ident}. If we were not to invoke these assumptions, we could not link differences in CATEs (HTEs) to differences in the AIEs of a mechanism. 

\subsection{Case 1: HTEs exist for outcome $Y(Z)$}

Proposition \ref{prop1} analyzes the case in which HTEs exist (are observed) and Assumptions \ref{cond:irre1} and \ref{cond:irre2} are assumed to hold for the outcome, $Y$. It shows that HTEs can provide evidence that a covariate is an MDV for some mechanism of interest, $M_1$. Recall that if $X_k$ is a MDV for mechanism 1, $AIE^Y_1(X_k = x)\neq AIE^Y_1(X_k = x')$ for some $x, x' \in X_k$. This is sufficient to provide evidence that mechanism $M$ is active for at least one unit.

\begin{prop}\label{prop1}Suppose Assumptions \ref{cond:irre1}-\ref{cond:irre2} hold with respect to $X_k$ for outcome $Y$. If HTEs exist with respect to $X_k$, then $X_k \in \textbf{X}^{MDV}$ for mechanism $M_1$.
\end{prop}

This conforms to standard interpretations that the HTEs provide evidence that a mechanism is active. Nevertheless, this finding relies critically upon the validity of exclusion assumptions, following Proposition \ref{prop:ident}. If one were to detect heterogeneity in an empirical study in this class, the heterogeneity would indeed provide evidence that the postulated mechanism, $M_1$, is active for some units.

\subsection{Case 2: HTEs do not exist for outcome $Y(Z)$}

We now consider the converse: the case when there exist no HTEs with respect to $X_k$ and Assumptions \ref{cond:irre1} and \ref{cond:irre2} are assumed to hold for the outcome, $Y$. 

\begin{prop}\label{prop2}Suppose Assumptions \ref{cond:irre1}-\ref{cond:irre2} hold with respect to $X_k$ for outcome $Y$. If no HTEs exist with respect to $X_k$, at least one of the following must be true:
\begin{enumerate}
    \item $X_k \notin \textbf{X}^{MDV}$ for mechanism $M_1$.
    \item  No MDV exists for mechanism $M_1$.
\end{enumerate}
\end{prop}

Proposition \ref{prop2} shows that a lack of HTEs provides less information with regard to mechanism activation than is generally asserted. Under the exclusion assumptions, there are two reasons why HTEs may not exist with respect to a covariate, $X_k$. First, it may be the case that $X_k$ is not a MDV for mechanism $M_1$. In this sense, we have misspecified the theoretical relationship between a given covariate and a mechanism. Second, it may be the case that no MDV exists for mechanism $M_1$. As we discuss in Corollary \ref{cor1}, there are two possible reasons why a MDV would not exist for mechanism $M_1$. Importantly, we show that this could happen with an active or an inert mechanism $M_1$.

\begin{cor}\label{cor1} If no MDV exists for a mechanism $M_1$, there are two possibilities:
    
    (1) Mechanism $M_1$ is not active. %(implied by the definition of MDV).
    
    (2) Mechanism $M_1$ is active, but produces the same effect for all units so there exists no $X_k$ for which $AIE^Y_1(X_k = x) \neq AIE^Y_1(X_k = x')$. 
\end{cor}

Case (1) of Corollary \ref{cor1} is implied by the definition of MDV. If a mechanism is inert---thereby producing an indirect effect of zero for all units---there cannot exist any MDVs, measured or unmeasured. In contrast, in Case (2), a mechanism can be active but it produces the same indirect effect for all units. Recall that a non-zero difference in average indirect effects at different levels of a covariate $X_k$ is a \emph{sufficient} condition for mechanism activation. However, it is not a \emph{necessary} condition for mechanism activation. These results show that, in contrast to standard interpretation, a \emph{lack} of heterogeneity cannot tell us about whether a mechanism is active. Moreover, our theory could be misspecified, meaning that our postulated MDV, $X_k$ is not actually a MDV. An assessment of HTEs with respect to a single moderator cannot distinguish between these three possibilities. Nor can we assign probabilities to these (non-mutually exclusive) explanations without stronger assumptions.

%Comparing Propositions \ref{prop1} and \ref{prop2}, the presence of HTEs provides more information with regard to mechanisms than does the absence of HTEs. In this sense, relying upon the a lack of heterogeneity to ``rule out'' a potential mechanism would require even stronger assumptions than we have imposed. Specifically, we would need to assume that $X_k \in \textbf{X}^{MDV}$ in order to rule out the activation of mechanism $M$. Indeed, such an assumption is precisely what we are trying to \emph{learn} from the presence of HTEs in Proposition \ref{prop1}. 

\subsection{Case 3: HTEs exist for transformed outcome $\widetilde{Y}(Z)$}

To understand what HTEs reveal with respect to a (non-linearly) transformed outcome $\widetilde{Y}(Z)$, it is useful to introduce one final concept.  We will denote $\textbf{X}^{R} = \{X_1,X_2,...,X_K\}$, as the set of all possible pre-treatment covariates with  non-zero effects on the outcome, $Y$.\footnote{Formally, if $X_k \in \textbf{X}^{R}$, then there exist $x \neq x' \in X_k$ such that $Y(Z, M(Z; X_k=x,X_{\neg k}); X_k=x, X_{\neg k}) \neq Y(Z, M(Z; X_k=x',X_{\neg k}); X_k=x', X_{\neg k})$.} Covariates in $\textbf{X}^R$ can be thought of as ``relevant'' for predicting outcome $Y$. It is also useful to let $\textbf{X}$ be the set of all possible pre-treatment covariates. It is clear that for any outcome, $Y$, and mechanism $M$, $\textbf{X}^{MDV} \subseteq \textbf{X}^R \subseteq \textbf{X}$. Typically, these subsets will be proper. 

We now return to our main question of interest: what do HTEs reveal with regard to mechanisms? Proposition \ref{prop3} considers the case when there are HTEs in a covariate $X_k$. Here, we can learn that $X_k \in \textbf{X}^R$, but this is not informative about whether $X \in \textbf{X}^{MDV}$, since $\textbf{X}^{MDV} \subseteq \textbf{X}^R$. Why do we observe heterogeneity in covariates in $\textbf{X}^R$ that are not MDVs? The non-linear transformation $h(\cdot)$ generically means that outcome $\widetilde{Y}$ does not satisfy Assumptions \ref{cond:irre1} and \ref{cond:irre2}, even if $Y$ satisfies both assumptions. In this case, the difference in CATEs no longer identifies a difference in the conditional AIEs of interest! Indeed, the difference in CATEs for the transformed outcome is produced by the (sum of) differences in conditional $ADE^{\widetilde{Y}}$'s and conditional $AIE^{\widetilde{Y}}$'s for all mechanisms. Violation of the identifying assumptions means that we can no longer attribute HTEs to the mechanism of interest, $M_1$. This means that (absent stronger assumptions), we cannot learn about the activation of $M_1$ from HTEs on the transformed outcome. 

\begin{prop}\label{prop3}
Suppose Assumptions \ref{cond:irre1}-\ref{cond:irre2} hold with respect to $X_k$ for outcome $Y$. Let observed outcome $\widetilde{Y}$ be a transformed outcome relative to $Y$. If HTEs exist with respect to $X_k$ for outcome $\widetilde{Y}$, then $X_k \in \textbf{X}^{R}.$
\end{prop}

\subsubsection{Case 4: HTEs do not exist for transformed outcome $\widetilde{Y}(Z)$}

We now return to a final case of our theoretical analysis by asking when we are examining a transformed outcome that may be affected by mechanism $M_1$, what can we learn from a \emph{lack} of HTEs? Proposition \ref{prop4} indicates that in this case, we can infer that $X_k \in \textbf{X}$. This is obviously a vacuous result. We already know that $X_k \in \textbf{X}$ since $X_k$ is a covariate and $\textbf{X}$ is the set of all covariates. However, the proposition tells us that we cannot assert that the absence of HTEs implies that $X_k$ is not a MDV. The logic here resembles the previous case. Non-linear transformation $h(\cdot)$ violates Assumptions \ref{cond:irre1} and \ref{cond:irre2}, which means that we cannot attribute a lack of heterogeneity in CATEs to a lack of differences in a given conditional AIEs. We purposely state a vacuous result to emphasize how little can be ascertained about mechanisms from the lack of HTEs when outcomes are transformed non-linearly. 

\begin{prop}\label{prop4}
Suppose Assumptions \ref{cond:irre1}-\ref{cond:irre2} hold with respect to $X_k$ for outcome $Y$. Let observed outcome $\widetilde{Y}$ be a transformed outcome relative to $Y$. If HTEs do not exist with respect to $X_k$ for outcome $\widetilde{Y}$, then $X_k \in \textbf{X}.$
\end{prop}

Often we make assumptions about the mapping $h$. For example, the mapping in \eqref{eq:likert} imposes assumptions about how latent attitudes translate into Likert-scale responses. When we are willing to make such assumptions, we can refine Proposition \ref{prop4} slightly. Specifically, in Proposition \ref{propa4}, we show that if Assumptions \ref{cond:irre1} and $\ref{cond:irre2}$ hold for an absolutely continuous $Y$, which is transformed by \eqref{eq:likert} (for any $Q \geq 2$ categories), if HTEs do not exist with respect to $X_k$, then $X_k \notin \textbf{X}^R$. Because $\textbf{X}^{MDV}\subseteq \textbf{X}^R$ we know then that $X_k \notin \textbf{X}^{MDV}$ if $X_k \notin \textbf{X}^R$. But as in Proposition \ref{prop2} and Corollary \ref{cor1}, there are multiple possible explanations: our theory about how $X_k$ relates to mechanism $M_1$ could be wrong or no MDV exists for mechanism $M_1$. These possibilities mean that we cannot make an inference about mechanism (non)-activation from the absence of HTEs with respect to $X_k$.

\subsection{Summary of results}

In sum, our propositions characterize four cases into which we can classify attempts to ascertain mechanism activation from HTEs, as described in Table \ref{tab:results}. This table suggests that in addition to the exclusion assumptions needed for identification of the difference of conditional average indirect effects of a mechanism (Assumptions \ref{cond:irre1} and \ref{cond:irre2}), HTEs serve as a useful test of the activation of a mechanism when evaluated for an outcome that satisfies these exclusion assumptions. However even with these outcomes, we can make an inference mechanism activation in the presence of heterogeneity, but not in the absence of heterogeneity. In the next section, we consider how these results should inform applied theories and research design. 

\begin{table}
\resizebox{\textwidth}{!}{
\begin{tabular}{cccp{11cm}} \hline
Case & Outcome type & HTEs? & What can be inferred about mechanism $M_1$? \\ \hline \hline 
1 & $Y$ & Yes & Difference in conditional indirect effect of mechanism $M_1$ is not zero $\Rightarrow$ Mechanism $M_1$ is \textbf{active} for at least one unit.\\ \hline
2 & $Y$& No & Three possible explanations (not mutually exclusive):\\
& & & \hspace{1em}(a) Covariate $X_k$ does not moderate the effect of mechanism $M_1$, regardless of whether the mechanism is active or not. \\ 
& & & \hspace{1em}(b) Mechanism $M_1$ is active and produces the same effect for all units. (No covariate moderates its indirect effect.) \\
& & & \hspace{1em}(c) Mechanism $M_1$ is inert (produces no effect) for all units. \\ 
&& & $\Rightarrow$ \textbf{No information about activation} of mechanism $M$ without further assumptions. \\ \hline 
3 & $\widetilde{Y}$ & Yes & Covariate $X_k$ predicts outcome $\widetilde{Y} \Rightarrow$ \textbf{No information about the activation} of mechanism $M_1$. \\ \hline
4 & $\widetilde{Y}$ & No & Without further assumptions, provides no information about relationship between $X_k$ and $\widetilde{Y}$. \textbf{No information about activation} of mechanism $M_1$.\\ \hline
\end{tabular}}
\caption{Summary of the interpretation of results. All cases assume that Assumptions \ref{cond:irre1} and \ref{cond:irre2} hold for covariate $X_k$, mechanism $M$, and outcome $Y$. Outcome $\widetilde{Y}$ is a transformed outcome relative to $Y$.}\label{tab:results}
\end{table}

\section{Applications and Guidance for Research Design}\label{sec:guidance}

How should these results guide future efforts to test mechanisms using heterogeneous treatment effects? In this section, we use a set of four published applications to illustrate how these considerations should inform interpretation and prospective research design. In contrast to many methodological papers, our focus is on how theorized mechanisms underlying treatment effects link to a specific estimand: a difference in CATEs.\footnote{These discussions generalize to a difference in conditional local average treatments (LATEs) or difference in conditional average treatment effects on the treated (CATTs).} There are multiple ways to estimate this difference. For a randomized binary treatment, $Z_i$ and binary moderator $X_{ik}$---a candidate MDV---the most straightforward estimator of this difference in CATEs is $\beta_3$ in the OLS regression equation:
\begin{align*}
Y_i &= \beta_0 + \beta_1 Z_i + \beta_2 X_{ik} + \beta_3 Z_i X_{ik} + \varepsilon_i
\end{align*}
For guidance on estimation, we refer readers to excellent treatments of estimation of HTEs (or interaction effects) including \citet{bramboretal2017,berryetal2009,hainmuelleretal20}. Additionally, see expositions of related inferential problems stemming from limited statistical power of interaction effects \citep[e.g.,][]{mcclellandjudd376} and the threat of multiple comparisons when HTEs in multiple moderators are estimated \citep[e.g.,][]{gerbergreen2012,leeshaikh2014,finketal2014}.

In line with our motivating example, we select the four applications that examine empirically partisan alignment (or bias) as a possible moderator for the effect of corruption information about an incumbent candidate. Appendix \ref{ss:context} describes the contexts and designs of these studies. Figure \ref{fig:dags} provides a summary of the theories invoked in each study. Panel (a)---which is analogous to our running model (while omitting valence)---suggests that information about corruption gives way to an analogous distaste for corruption mechanism. But voters also value the partisanship of the incumbent. Distaste for corruption (the mechanism) and partisanship enter voter preferences for the incumbent as distinct and additively separable considerations of voters. Panel (b) is equivalent to Panel (a), but includes a second mechanism linking corruption information to voter utility: a motivation of voters to coordinate votes within their precinct. Panel (c) returns to a single mechanism, voter distaste for corruption, but suggests that partisan alignment also conditions how voters process corruption information, suggesting the presence of motivated reasoning. Finally, Panel (d) includes only voter distaste for corruption and suggests that the degree to which voters are averse to corruption is a function of their partisan leanings. We justify our representation of each cited theoretical account in Appendix \ref{ss:theory}. We note that the three papers with a single mechanism largely focus on voter distaste for observed corruption and emphasize the role of partisanship in moderating this effect. In contrast, \citet{ariasetal2019} focus on the second mechanism---voter coordination---whereas we will focus (largely) on the distaste mechanism in our analysis. We note that, unlike our model, these papers do not focus on corruption aversion as a moderator, though such variation is generally consistent with the theories they articulate (see Appendix \ref{ss:theory}).

\begin{figure}
\begin{subfigure}[t]{.48\linewidth}
\begin{center}
\begin{tikzpicture}
\node (z) at (1.5, .5) {$c$}; 
\node (lambda) at (2.25, 1.5) {$\lambda$};
\node (m) at(3, .5) {$M$};
\node (a) at (1.5,-.5) {$a$};
\node (y1) at (5, 0) {$y_1$}; 
\node (y2) at (6, 0) {$y_2$};
\draw[->] (z)--(m);
\draw[->] (m)--(y1);
\draw[->] (y1)--(y2);
\draw[->] (a)--(y1);
\draw[->] (lambda)--(2.25,.65);
\end{tikzpicture}
\end{center}
\caption{Distaste for corruption ($M$) and partisan alignment ($a$) are distinct and additively separable arguments in voters' preferences for the incumbent ($y_1$), as in the running example. \\ \textbf{Source:} \citet{eggers2014}}
\end{subfigure}
  \hspace{.5em}
  \begin{subfigure}[t]{.48\linewidth}
  \begin{center}
\begin{tikzpicture}
\node (z) at (.5, .5) {$c$}; 
\node (lambda) at (1.75, 1.75) {$\lambda$};
\node (m) at(3, 1) {$M$};
\node (r) at (3, 0){$R$};
\node (a) at (1.5,-1.5) {$a$};
\node (y1) at (5, 0) {$y_1$}; 
\node (y2) at (6, 0) {$y_2$};
\draw[->] (z)--(m);
\draw[->] (m)--(y1);
\draw[->] (z)--(r);
\draw[->] (r)--(y1);
\draw[->] (y1)--(y2);
\draw[->] (a)--(y1);
\draw[->] (lambda)--(1.75, .8);
\end{tikzpicture}
\end{center}
\caption{There are two mechanisms, distaste for corruption ($M$) and voter coordination motives ($R$). Partisan alignment ($a$) is distinct from both mechanisms and additively separable arguments in voters' preferences for the incumbent ($y_1$), as in the running example. \\ \textbf{Source:} \citet{ariasetal2019}}
\end{subfigure}
\par\bigskip
\par\bigskip
\begin{subfigure}[t]{.48\linewidth}
\begin{center}
\begin{tikzpicture}
\node (z) at (1.5, .5) {$c$}; 
\node (lambda) at (2.25, 1.5) {$\lambda$};
\node (m) at(3, .5) {$M$};
\node (a) at (1.5,-.5) {$a$};
\node (y1) at (5, 0) {$y_1$}; 
\node (y2) at (6, 0) {$y_2$};
\draw[->] (z)--(m);
\draw[->] (m)--(y1);
\draw[->] (y1)--(y2);
\draw[->] (a)--(y1);
\draw[->] (lambda)--(2.25, .6);
\draw[->] (a)--(2.25, .4);
\end{tikzpicture}
\end{center}
\caption{Partisan alignment (bias) affects voter processing of corruption information via motivated reasoning and enters voters' preferences for the incumbent through a distinct channel.\\ \textbf{Source}: \citet{anduizaetal2013}}
\end{subfigure}
  \hspace{.5em}
\begin{subfigure}[t]{.48\linewidth}
\begin{center}
\begin{tikzpicture}
\node (z) at (1.5, -.5) {$c$}; 
\node (lambda) at (2.25, .5) {$\lambda$};
\node (m) at(3, -.5) {$M$};
\node (a) at (1.5,1.5) {$a$};
\node (y1) at (5, 0) {$y_1$}; 
\node (y2) at (6, 0) {$y_2$};
\draw[->] (z)--(m);
\draw[->] (m)--(y1);
\draw[->] (y1)--(y2);
\draw[->] (a)--(y1);
\draw[->] (lambda)--(2.25, -.4);
\draw[->] (a)--(lambda);
\end{tikzpicture}
\end{center}
\caption{Partisan alignment ($a$) conditions corruption aversion ($\lambda$), which moderates distaste for corruption ($M$). Alignment also enters voter's preferences for the incumbent ($y_1$) through a distinct channel. \\ \textbf{Source}:  \citet{defigueredoetal2023}}
\end{subfigure}
\caption{Four theoretical accounts of how partisan alignment (or bias) and information relate to voter preferences for the incumbent and vote choice. $c$ refers to corruption information; $\lambda$ is a voter's corruption aversion; $M$ is voters' distaste for observed corruption; $a$ is voters' partisan bias toward the incumbent; $y_1$ is voter utility (preference) for the incumbent; and $y_2$ is vote choice for the incumbent.}\label{fig:dags}
\end{figure}

Our framework suggests four avenues for improved use of HTEs for mechanism detection which we explore below. First, it suggests a set of necessary theoretical considerations/arguments. Second, it suggests improvements in the interpretation of observed HTEs. Third, it offers recommendations for prospective research design. Finally, it provides guidance on when stronger theoretical assumptions may be useful to impose. 

\subsection{Three essential theoretical considerations} 

Our framework identifies three attributes of a theory that are needed to support any analysis of causal mechanisms using HTE. First, researchers must \textbf{identify a set of candidate mechanisms}. The examples we cite are clear in positing a set of candidate mechanisms. Three---Panels (a), (c), and (d)---focus on a single mechanism: distaste for an incumbent's corruption. One---Panel (b)---instead suggests that distaste and voter coordination motives are two distinct mechanisms through which corruption information affects voter utility. The applications we discuss are quite clear about the mechanisms thought to underlie the effects that they measure. 

Second, in order to use HTE to learn about a given mechanism, however, our analysis shows that one must (1) \textbf{invoke exclusion assumptions} and (2) \textbf{identify candidate MDVs}. In theories with a single mechanism---as in Panels (a), (c), and (d) of Figure \ref{fig:dags}---it is only necessary to defend Assumption \ref{cond:irre1} for a given covariate. In the case of the partisan alignment covariate, this means that we should justify the \emph{absence} of a direct effect of corruption information that varies in partisan alignment. In Panel (b), when there are multiple mechanisms, one would need to invoke Assumptions \ref{cond:irre1} and \ref{cond:irre2} for a given covariate, e.g., partisan alignment. As represented in the papers (and therefore in Figure \ref{fig:dags}), the exclusion assumptions should hold for the voter utility outcome ($y_1$) in each case. We note that because these studies focus on a small (and well-articulated) set of mechanisms, assessing whether Assumptions \ref{cond:irre1} and (where relevant) \ref{cond:irre2} are plausible is relatively straightforward. In other types of studies that propose a larger set of candidate mechanisms, (1) additional exclusion assumptions must be invoked for each candidate mechanism; and (2) justification of these exclusion assumptions becomes more difficult because there are more plausible violations of these assumptions.

The different theories vary in whether partisan alignment, $a$, should serve as an MDV for voter utility, $y_1$. In the first two theories, depicted in Panels (a) and (b), partisan alignment is \emph{not} a candidate MDV. We can see this because there is no arrow from $a$ to the edge between $c$ and $M$ or the edge between $M$ and $y_1$, or (in Panel (b)), the voter coordination mechanism ($R$). In contrast, $a$ serves as an MDV for the distaste mechanism in Panels (c) and (d). We can see this from the direct arrow from $a$ to the edge between $c$ and $M$ in Panel (c)---which captures voters' motivated reasoning about the corruption disclosure---and from $a$ to $\lambda$ to the edge between $c$ and $M$ in Panel (d), which captures the idea that partisan affiliation shapes corruption aversion, which in turn conditions the degree to which corruption information ($c$) gnerates a distaste for observed corruption. 

Third, authors must consider the \textbf{relationship between outputs of the candidate mechanisms and measured outcomes}. The empirical studies associated with Panels (a), (b), and the field experiment in (d) measure aggregate vote share at the level of the constituency (Panel a) or precinct (Panels b and d). Within our framework, aggregate vote share can be thought of as $\sum_i y_{i2}/n$, where $n$ is the number of voters.\footnote{In Panel (a), \citet{eggers2014} also provides a survey-based measure of $y_2$ from the British Election Survey by analyzing self-reported vote choice. Some studies also consider vote share as a fraction of registered voters. This would require a model that allows voters to abstain, but yields substantively similar results.} This means that we should expect HTEs in partisan alignment for each of these studies. But recall that of these three studies, partisan alignment is only an MDV under the theory in (d) in which alignment conditions a voter's corruption aversion. In contrast, the survey-based studies in Panel (c) and the survey experiment component of Panel (d) measure outcomes that resemble voter utility through Likert scales. In both relevant Panels, the authors expect partisan affiliation to moderate the effect of corruption information on distaste for corruption (albeit for different reasons). So in contrast to the our model (and the examples in Panels (a) and (b)), we would expect to observe HTE in partisan affiliation for these outcomes if the distaste for corruption mechanism is active. 

\subsection{Improving the interpretation of HTEs as mechanism tests} 

By summarizing our results, Table \ref{tab:results} provides a guide for interpreting the presence or absence of HTEs as tests of a mechanism. Importantly, they point to an asymmetry in what we can infer about mechanisms from the presence versus absence of HTEs. Specifically, when exclusion assumptions hold, the presence of HTEs provides evidence that a mechanism is active. In contrast, the absence of HTEs does not rule out a mechanism or show that it is inert without further assumptions.

Our results are theoretical or identification results, meaning that they can be interpreted as what would obtain if we had an infinite sample. But, in practice, empiricists operate in a world with finite---and often relatively small---samples. This introduces statistical problems as well. In the context of HTEs, known limits to the statistical power for interaction effects (or terms) reduces our ability to detect HTEs that do exist. In other words, we risk many false-negatives in inferences related to the existence of heterogeneity. Because a lack of HTEs provides less informative about mechanism activation, low power suggests that applied researchers often operate in a world in which heterogeneity analysis is unlikely to provide information to support inferences about mechanisms.\footnote{Selective reporting of significant results complicates the situation further. In this case, evidence in favor of treatment-effect heterogeneity is more likely to be a false-positive, which increases the the probability that researchers infer that a mechanism is active when it is not.} 

In the context of the studies that we examine, authors are admirably transparent about these limitations. For example, \citet[][p. 735]{defigueredoetal2023} (Panel (d)) write ``given the small samples, however, the difference between the two [CATE] estimates (the interaction) is not statistically significant. Still, the difference in magnitudes certainly suggests that [candidate's] voters are more sensitive to corruption-related information than supporters of other candidates.'' In sum, the combination of the asymmetry in the inferences about mechanisms with versus without HTEs combined with the limited ability to detect heterogeneity should give authors caution in the interpretation of HTEs as tests of mechanisms.

\subsection{Guidance for prospective research design}

Our framework posits several recommendations for the design of causal research that seeks to test mechanisms quantitatively using HTE. Our suggestions are premised on improvements in measurement. 

\textbf{Measure more candidate MDVs}: In terms of covariates, we are primarily concerned with which covariates are measured and the number of candidate MDVs (per mechanism) among those covariates. Following the exclusion assumptions, covariates are only useful for ascertaining mechanisms when (1) they are plausibly MDVs for a single mechanism; and (2) they do not moderate direct effects. This observation suggests that special care must be taken when positing candidate MDVs. When pre-treatment covariates are (largely) collected in baseline data collection, there is a need to posit MDVs and defend exclusion assumptions \emph{ex-ante}. Such considerations require more theory and justification than are typically conveyed in the specification of moderation or HTE analyses in pre-analysis plans. 

Further, it is very useful to have multiple candidate MDVs for a given mechanism. To see why, consider the case in which we have two candidate MDVs, $X_1$ and $X_2$ for mechanism $M_1$ and the exclusion assumptions hold for both candidate MDVs. Suppose that there do not exist HTEs in $X_1$ but there do exist HTEs in $X_2$. If we only measured HTEs with respect to $X_1$, following Proposition \ref{prop2}, we would not be able to ascertain whether the problem is with the theory ($X_1 \notin \textbf{X}^{MDV}$) or whether there simply exist no MDV for mechanism $M_1$. If there exist HTEs in $X_2$, we can eliminate the possibility that there do not exist MDV for mechanism $M_1$. This would suggest that the theory with respect to $X_1$ is misspecified. This is useful insofar as it allows us to make an inference that mechanism $M_1$ is active. Note, however, that in order to leverage multiple candidate MDVs, the exclusion assumption must hold for each candidate MDV, which can be quite demanding. 

The simplified presentation of the theories in Figure \ref{fig:dags} are unlikely to reflect the full set of candidate MDVs for a given mechanism. However, we can explore this logic with respect to the theory in Panel (c) where there are two candidate MDVs. Specifically, in this study, corruption aversion and partisanship are posited to be MDVs for the distaste for corruption mechanism. \citet{anduizaetal2013} measure only the partisanship indicator and find that respondent distaste for a (hypothetical) incumbent's corruption varies in partisanship. Suppose instead that the authors they had also measured corruption aversion---and detected HTEs in that variable---\emph{while} failing to detect HTEs in partisanship. Such a finding would suggest that the distaste mechanism is active (via the HTEs in corruption aversion). It would additionally indicate that partisanship may not be a MDV for the distaste mechanism, perhaps casting doubt on the role of motivated reasoning in shaping a voters' distaste. 

\textbf{Prioritize specific outcomes for mechanism tests}: Our focus on non-linear transformations of measured outcomes yields a further recommendations for research design. If a goal of a research design is to destinguish between mechanisms or detect a posited mechanism, outcomes for which Assumptions \ref{cond:irre1} and \ref{cond:irre2} are plausible should be prioritized in HTE analysis. For example, if researchers had survey measures of voter utility from the incumbent and vote choice (from survey responses of administrative vote returns at the precinct level), any inference about the mechanism should be made on the basis of the utility measure. The logic behind this choice helps to convey the structure of the theorized mechanism and its relationship to the outcome measures. Moreover, ex-ante specification of the outcomes for which these assumptions are likely to hold can guide choices about which outcome variables to invest in measuring. This can also guide pre-analysis plans.

To be clear, our recommendation is \emph{not} to avoid the estimation of HTEs for non-linearly transformed outcomes entirely. Indeed, in the study of elections, we typically care about vote choice more than voter utility, since this vote choice determines who wins office. But we should be clear about why HTEs on vote choice matter when they are not informative about mechanisms. They could be used to understand how to better target future corruption information interventions across the electorate \citep[e.g.,][]{atheywager2021,kitagawatetenov2018}, to extrapolate effects to other electorates under a model \citep[e.g.,][]{egamihartman2020}, or to describe about the distributional impacts of the treatment. These are all valid---and perhaps underutilized---uses of HTEs which do not depend on the relationship between mechanisms and measured outcomes. We simply advise researchers to more carefully articulate the goals of their analysis. 

\subsection{Imposing assumptions about to use HTEs as mechanism tests}\label{ss:assumptions}
 
Our results suggest that the presence of HTEs is not informative of mechanism activation when the exclusion assumptions do not hold. Model-based approaches may be useful when researchers only have access to transformed outcomes or have reasons to doubt the validity of the exclusion assumptions. For example, for the field experimental and observational studies in Figure \ref{fig:dags} (Panels (a), (b), and the field experiment in (d)), authors have vote share outcomes or self-reported vote choice, but recall that this is a transformed outcome. Furthermore, in Panel (b), it may also be reasonable to believe that corruption aversion conditions both the distaste for corruption and voter coordination mechanisms. Can we make progress in these settings by imposing stronger assumptions about the mapping from voter preferences (utility) to vote choice? 

We consider the possibility of invoking different models or sets of assumptions may aid in using HTEs to learn about mechanisms. Table \ref{tab:modelbased} summarizes three problems identified by our analysis. Each problem is paired with a statistical model or set of assumptions that we examine as a solution. The right column summarizes the conclusions of our analyses, which are detailed in greater depth in \ref{app:sol}. In sum, two of the three models/sets of added assumptions are strong enough (in isolation) to allow HTE to provide information about mechanism activation, where they do not in the absence of these additional assumptions.

\begin{table}\resizebox{\textwidth}{!}{
\begin{tabular}{lp{5cm}p{5cm}p{7cm}} \hline
&Problem & Model/assumptions & Conclusions\\ \hline\hline 
1 &Exclusion assumptions fail because utility is transformed into a discrete choice. & Random utlity model specifies a systematic and random component of utility that generates observed outcomes.& In settings in which mechanisms operate on utility but we observe choice, random-utility models can recover a model-based estimate of expected utility that can be used for analysis of mechanism activation through HTEs. \\  \hline 
2 & Exclusion assumptions fail because outcome is transformed. & Assumption about monotonicity of CATEs in a candidate MDV $X_k$. & Monotonicity is not strong enough to provide information about mechanism activation without ancillary functional form assumptions. \\ \hline 
3 & Exclusion assumptions fail because covariate may moderate multiple mechanisms. & Bayesian model of CATE magnitude under different mechanism activation profiles. & Plausible when we are willing to register priors about relative CATE magnitudes under different mechanisms. \\  \hline 
\end{tabular}}
\caption{Summary of model- or assumption-based alternatives to exclusion assumptions. }\label{tab:modelbased}
\end{table}

\textbf{Modelling non-transformed outcomes/random utility models}: In the context of transformations from utility to choice---like the transformation of voter utility from the incumbent into a vote for the incumbent---there exist widely-used random utility models that seek to recover preferences (utility) from choices. These models provide a functional mapping between an individual's utility (the outcome upon which the mechanism operates) and their choice by decomposing vote choice into an observed systematic and an unobserved random component. In our motivating model, the information treatment, the corruption aversion covariate, and a partisan alignment covariate are systematic and could, in principle, be observed. In contrast, the valence shock is random. In Panel (b) of Figure \ref{fig:dags}, corruption information (the treatment), corruption aversion, and the network structure of a precinct (or a sufficient statistic thereof) could be observed, but partisan alignment is random. By specifying the systematic component as a function of individual- or choice-specific covariates and assuming a distribution of the random component(s), researchers may be able to estimate the systematic component of utility. 

Appendix \ref{app:transform} analyzes the use of random-utility models in the context of HTE analysis. These models have two principal merits in the present context. First, the parameterization of the systematic component allows researchers to assess whether the exclusion assumptions are plausible under a given theoretical model. If these assumptions are plausible, the second benefit of a random utility model emerges. Specifically, it permits researchers to estimate (or ``back out'') an outcome for which HTEs can provide information about mechanism activation, specifically $\mathbb{E}[Y]$.  Importantly, the invocation of a random-utility model is not free: it makes strong parametric assumptions about utility and its relationship to choice outcomes. Researchers may not be well-positioned to invoke or assess these assumptions. However, these assumptions allow for more formal examination of exclusion assumptions and may yield information which may permit learning about mechanism activation from HTEs. 

\label{r2:3}While random utility models are the best established method for estimating actors' utility from choice, the broader approach could be useful for other types of transformed outcomes. Suppose that empirical researcher believed that she observed $\widetilde{Y} = h(Y)$ and that the exclusion assumptions held for $Y$ (but not $\widetilde{Y})$. If she were willing to impose some (invertible) functional form on $h(\cdot)$, it may be possible to evaluate (or estimate) $h^{-1}(\widetilde{Y})$. As in the case of a random utility model, she could then conduct analysis with $E[Y]$.%\footnote{For non-deterministic functions, we can represent the transformed outcome as $\widetilde{Y} = h(Y, \epsilon)$.  If the researcher were willing to impose some functional assumption on $h(\cdot)$ and the distribution of $\epsilon$, it may be possible to evaluate or estimate (the linear function of) $Y$. } 

\textbf{Imposing monotonicity}: To this point, we analyze the current practice of examining the presence of HTEs as a test of mechanisms. However, examining the magnitudes of estimated CATEs may offer more information. We first consider whether the invocation of \emph{monotonicity of CATEs}---a common empirical assumption \citep{manski1997}---can provide evidence about mechanisms. In our context, monotonicity holds that for all $x'>x \in X_k$, $CATE(x') \geq (\leq) CATE(x)$. In the case of our model (and the applications in Panels a and b), corruption information should have a stronger (negative) effect on voter utility among voters with stronger corruption aversion (larger $\lambda$). Indeed, Remark \ref{remark1} shows that for $\lambda > \lambda'$, $|CATE(y_1 | X_1 = \lambda)| > |CATE(y_1 | X_1 = \lambda')|$ for the voter utility outcome. Is it possible that this monotonicity in $\lambda$ is maintained for the vote choice outcome? If this were the case, we could simply compare the magnitude of effects at different levels of the MDV ($\lambda)$ for some suggestive evidence about mechanism activation. Similarly, could a violation of monotonicity of the form $\text{sign}(CATE(y_2 | \lambda)) \neq \text{sign}(CATE(y_2 | \lambda'))$ provide evidence against mechanism activation? 

Unfortunately, in Appendix \ref{app:monotonicty}, we show that assuming monotonicity is not sufficient to provide information about mechanism activation through analysis of CATE magnitudes. Specifically, in Proposition \ref{prop:monotone}, we show that montonicity alone is not sufficient to ensure that HTEs take different signs when $X_k$ is not a MDV. Moreover, it does not ensure that CATEs are monotonic in $X_k$ even when $X_k$ is a MDV and montonicity is satisfied for the non-transformed outcome. These results show that we would need additional parametric assumptions for monotonicity to provide sufficient information to distinguish mechanism activation.

%Of course, these results consider only one data generating process when considering the implications of assuming monotonicity. But a single example reveals this class of assumptions cannot, in isolation, allow us to distinguish mechanisms through HTE. If researchers are willing to specify a functional form for directly-affected outcome $Y$ and mapping $h(Y)$, it may well be possible to generate a set of sufficient conditions for learning about mechanisms through HTEs on the outcome $h(Y)$. But this requires far stronger assumptions than monotonicity. 

\textbf{Learning from the Magnitude of HTEs}: The magnitude of CATEs may provide additional information than the presence of HTE. For example, suppose that---in contrast to Panel (b) of Figure \ref{fig:dags}---theory suggested that the corruption aversion moderator was a candidate MDV for both the distaste and coordination mechanisms. If this were the case, the exclusion assumptions (Assumptions \ref{cond:irre1} and \ref{cond:irre2}) would not be satisfied. However, we may have reasons to believe that the HTEs coming from one mechanism are large (in magnitude) while those coming from the other mechanism are small (in magnitude). In Appendix \ref{sec:magnitude}, we propose a simple Bayesian model that allows for inference about the activation of a specified mechanism from the magnitude of estimated effects. This model requires specification of priors over: (1) the probability of activation of a given mechanism and (2) the magnitude of HTEs under each candidate mechanism. Specification of these priors constitutes the invocation of additional assumptions.

We note that most theories in the social sciences admit directional---rather than point---predictions. By moving from the existence of HTE to their magnitude in this Bayesian setting, the model that we propose requires researchers to specify priors about the \emph{size} of HTEs, not simply their existence or direction. While these priors allow us to glean some information about mechanisms when the exclusion assumption(s) are violated, more work is needed to guide researchers in specifying such priors from theories that are largely directional.

\section{Conclusion}

Social scientists routinely estimate HTEs with the goal of understanding which mechanisms generate treatment effects. By providing the first theoretical analysis of the relationship between HTEs and mechanisms, we show that detecting mechanisms with HTEs is far less straightforward than implied by current practice. Specifically, any link between a covariate (moderator) and a mechanism requires exclusion assumptions, so that covariate does not moderate the effects of other mechanisms (or the direct effect). Even when these assumptions hold, we can only use HTEs to affirm the activation of a mechanism when (1) HTEs exist and (2) for transformed outcomes that preserve the additive separability of a mechanism's effect. Outside this case, HTEs do not provide sufficient information to show that a mechanism is active or inactive. In this sense, HTEs analysis should not be used to rule out activation of a mechanism without stronger assumptions than those that we impose. 

At present, mechanism detection is the modal use of HTEs in political science (see Table \ref{tab:class}) and the modal method for mechanism detection \citep{blackwelletal2024}. However, mechanism detection is not the only use of HTE. Our results speak to contexts where mechanistic analysis is a goal. HTEs are also increasingly used for extrapolation of treatment effects to different populations/settings \citep{egamihartman2020,devauxegami2022} and the targeting of treatments \citep{atheyetal2019,kitagawatetenov2018}. Our work does not directly speak to these uses of HTEs, because these methods do not seek to attribute observed effects to mechanisms \citep{sloughtyson2023}. 

Our analysis raises a number of issues and opportunities for future research to build upon. In particular, we emphasize that choices about which outcomes we measure can complicate efforts to understand the substantive mechanisms at work. For example, even if the exclusion assumptions hold for one outcome, by imposing a common non-linear transformation on that outcome, the exclusion assumptions can be violated. This distinction has underappreciated implications for multiple quantitative methods to detect mechanism activation, including mediation analysis, analysis of treatment effects on intermediate outcomes, and efforts to link the sign of treatment effect to a (set of) mechanism(s). One potentially fruitful avenue for continued use of HTEs for mechanism detection would be to move from the presence of HTEs to their magnitude, as we outline in Appendix \ref{sec:magnitude}. Adoption of this approach will rely on the the development of closer links between theoretical mechanisms and the size of reduced-form treatment effects than is current practice. 

\clearpage

\putbib % use your .bib file name here
\end{bibunit}

\newpage

%%%%%%%%%%%%%%%%%%%%%%%

\clearpage
\setcounter{page}{1}

% set appendix toc
\appendix
\addcontentsline{toc}{section}{Appendix} % Add the appendix text to the document TOC
\part{Appendix} % Start the appendix part
\parttoc % Insert the appendix TOC

\setcounter{figure}{0}
\setcounter{table}{0}
\renewcommand\thefigure{A.\arabic{figure}}
\renewcommand\thetable{A.\arabic{table}}

\onehalfspacing
\setcounter{page}{1}

\clearpage

\section[Additional Classification of Articles]{Additional Classification of Articles}

Table \ref{tab:class2} provides an additional classification of the articles described in Table \ref{tab:class} but disaggregates by research design. Note that we collapse difference-in-difference and panel analyses into one category that includes two-way fixed-effects and other estimators of the average treatment effect on the treated (ATT). We also collapse IV and natural experimental analyses into a single category that includes studies with some claim of exogenous variation not created by researchers that is argued to facilitate identification of an average treatment effect (ATE); an intent-to-treat effect (ITT); or a local average treatment effect (LATE).\footnote{This LATE is often termed the complier average causal effect (CACE) in political science.} This table shows that using HTEs to detect mechanisms is not unique to any one research design in common usage; the proportions of articles  in these journals that uses HTEs as a mechanism test (given by the ``weighted average'' column) is quite similar across all of these designs. 

\begin{table}[ht]
\centering
		\begin{tabular}{l|ccc|ccc|c} \hline
		 & \multicolumn{3}{c|}{}&  \multicolumn{3}{c|}{Pr(Reports HTEs as} & \\
			
			 & \multicolumn{3}{c|}{Total }&  \multicolumn{3}{c|}{mechanism test)} & Weighted \\
				Research design & \emph{AJPS} &  \emph{APSR} & \emph{JoP} & \emph{AJPS} & \emph{APSR}  & \emph{JoP} & average\\ \hline
				Experiment & 14 & 21&32& 0.50&0.43&0.53 & 0.49\\
				Difference-in-differences or panel &9&10&14&0.44&0.40& 0.36 & 0.39\\
				Regression discontinuity &2 &2&5&0.50&1.00&0.40 & 0.55\\
				IV or natural experiment&2 &5&7&0&0.80&0.71 & 0.64 \\
				Selection on observables &7&26&41&0.71&0.65&0.44 & 0.54\\ \hline
	\end{tabular}
\caption{Authors' classification of articles published in three leading political science journals in 2021 by research design. Note that the probabilities reported are those implied by Pr(Reports HTE) $\times$ Pr(Mechanism Test$\mid$Reports HTE) in Table \ref{tab:class}. In this table, we do not include quantitative articles without an apparent mapping to a (reduced-form) causal estimand. These omitted articles employ empirical research designs including structural estimation, development of new measures, and claims to measurement of correlations alone.}\label{tab:class2}
\end{table}

\section[Examples of Current Practice]{Examples of Current Practice}

As reported in Table \ref{tab:class}, 82\% of the articles that report HTEs interpret these quantities as tests of a mechanism. Two interpretations of HTEs are common. First, the \emph{presence} of HTEs with respect to a specific covariate provides evidence that a mechanism is active. For example, \citet{haimetal2021} report the results of a field experiment in a conflict-affect region of the Philippines that randomized provision of a program that connected villages to state services and sought to increase village leaders' trust in the state. In the context of the COVID-19 pandemic, they show that the program increased village leaders' probability of reporting COVID risk information to a government task force by $\approx 10$ percentage points. \citet{haimetal2021} argue that the program increased reporting by increasing beliefs in government competence. They report HTEs that indicate the effect on reporting was significantly larger in communities in which village leaders initially believed that the government ``[did] not have capacity to meet needs.'' These HTEs with respect to perceived capacity suggest that leaders assigned to treatment updated their beliefs about government capacity (the mechanism), thereby increasing their willingness to report risk information to the government. \\

Second, the \emph{absence} of HTEs with respect to a given covariate is frequently used to ``rule out'' the activation of a possible mechanism (often called an alternative explanation). For example, \citet{moscowitz2021} argues that the proportion of in-state residents in one's local media market increases rates of voter political knowledge and split-ticket voting in the US, countering trends toward the nationalization of politics. The paper provides evidence that news coverage is the mechanism that drives this effect. However, it also seeks to rule out an alternative mechanism: campaign advertising. The logic for this alternative mechanism holds that more in-state residents leads to more campaign advertising about in-state candidates, which in turn increases voter knowledge. Since advertisements air when during election season (when incumbents contest re-election), if the advertising mechanism were operative, we would expect the effect of in-state residents on voter knowledge to be larger during election seasons than other times in the electoral cycle. \citet{moscowitz2021} finds no such heterogeneity. The absence of heterogeneity is used to provide evidence against the advertisement-based mechanism.  Both of the above examples are exceptionally clear in delineating the mechanisms under consideration using HTEs, and therefore serve as exemplars of current practice.

\section[Motivating Example]{Motivating Example}
\subsection{Data Generating Process}

Figure \ref{si:fig_dag} depicts the data generating process that is evaluated by the HTE analysis in Remarks \ref{remark1}-\ref{remark2}. In the model, only the distaste for corruption mechanism is active. This mechanism is evaluated by examining heterogeneity in treatment effects with respect to $\lambda$. It is clear that the partisan alignment does not affect the distaste for corruption mechanism. \\ 

\begin{figure}[h]
\begin{tikzpicture}
\node (z) at (1.5, .5) {$c$: Corruption information}; 
\node (lambda) at (5, 1.5) {$\lambda$: Corruption aversion};
\node[text width = 6em] (m) at(7, .5) {$M$: Distaste for corruption};
\node (a) at (5,-.5) {$a$: Partisan alignment};
\node (y1) at (10, 0) {$y_1$: Utility}; 
\node (y2) at (13, 0) {$y_2$: Vote choice};
\draw[->] (z)--(m);
\draw[->] (m)--(y1);
\draw[->] (y1)--(y2);
\draw[->] (a)--(y1);
\draw[->] (lambda)--(5, .6);
\end{tikzpicture}

\caption{Data generating process of the model of corruption revelation and behavior. The arrow from the $\lambda$ to the path between the information and the mediator indicates that the $\lambda$  moderates the distaste for corruption mechanism.}\label{si:fig_dag}
\end{figure}

\subsection{Proofs of Remarks \ref{remark1}-\ref{remark2}}

\subsubsection*{Remark 1(a) : }

\begin{proof}
	$$
	\begin{aligned}
		CATE(y_1, X_1= \lambda)&= E[y_{i1}|c_i =1,\lambda]-E[y_{i1}|c_i=0,\lambda]\\
		&=E[-\lambda + a_i -\varepsilon_i]-E[a_i -\varepsilon_i] \\ 
		&=-\lambda\\
	\end{aligned}
	$$
	and 
	$$
	\begin{aligned}
		CATE(y_1, X_1 = \lambda')&= E[y_{i1}|c_i=1,\lambda]-E[y_{i1}|c_i=0,\lambda]\\
		&=E[-\lambda' + a_i -\varepsilon_i]-E[a_i -\varepsilon_i] \\ 
		&=-\lambda'
	\end{aligned}
	$$
	Because $\lambda > \lambda'$, we conclude $|CATE(y_1, X_1=\lambda)| > |CATE(y_1, X_1=\lambda')|.$

\end{proof}

\subsubsection*{Remark 1(b) :} 

\begin{proof}
	
	$$
	\begin{aligned}
		CATE(y_1, a)&=
		E[y_{i1}|c_i=1,a]-E[y_{i1}|c_i=0,a]\\
		&=E[-\lambda_i + a -\varepsilon_i]-E[a -\varepsilon_i]\\
		&=-\mathbb{E}[\lambda_i] 
	\end{aligned} 
	$$
	
	$$
	\begin{aligned}
		CATE(y_1, a')&=
		E[y_{i1}|c_i=1,a']-E[y_{i1}|c_i=0,a']\\
		&=E[-\lambda_i + a' -\varepsilon_i]-E[a' -\varepsilon_i]\\
		&=-\mathbb{E}[\lambda_i] 
	\end{aligned} 
	$$
	As a result, we conclude $CATE(y_1, X_2=a) = CATE(y_1, X_2=a').$
	
\end{proof}

\subsubsection*{Remark 1(c) :} 

Follows directly from Remark 1(a) and 1(b).

\subsubsection*{Remark 2(a) :} 

\begin{proof}
	
	Recall that $y_2$ is given by:
	$$
	\begin{aligned}
		y_2 &= \begin{cases}
			1 &\text{ if } -\lambda_i c_i + a_i + \varepsilon_i \geq 0\\
			0 &\text{ else }
		\end{cases}
	\end{aligned}
	$$
	$CATE(y_2, X_1 = \lambda)$ is therefore given by:
	$$
	\begin{aligned}
		CATE(y_2, X_1 = \lambda)&=E_{a_i,\varepsilon_i}[y_{i2}|\lambda,c_i=1]-E_{a_i,\varepsilon_i}[y_{i2}|\lambda,c_i=0]\\
		&= \mathbb{E}_{a_i}[\Pr(\varepsilon_i \ge \lambda - a_i)]- \mathbb{E}_{a_i}[\Pr (\varepsilon_i \ge -a_i)]\\
		&= \sum_{a_i \in \{-1,0,1\}} \frac{1}{3} [\Phi(-a_i)- \Phi(\lambda -a_i)]
	\end{aligned}
	$$ the second line follows the independence of $\varepsilon$, and $\Phi$ is the standard normal CDF.
	
	Similarly,
	
	$$
	\begin{aligned}
		CATE(y_2, X_1 = \lambda')&=\sum_{a_i \in \{-1,0,1\}} \frac{1}{3} [\Phi(-a_i)- \Phi(\lambda' -a_i)]\\
	\end{aligned}
	$$
	
	Therefore, $CATE(y_2, X_1 = \lambda)-CATE(y_2, X_1 = \lambda') =  \frac{1}{3} \sum_{a_i \in \{-1,0,1\}}[\Phi(\lambda' -a_i)-\Phi(\lambda -a_i)] < 0$ when $\lambda>\lambda'$. Recall the treatment effect is negative (for example from the proof of Remark 1(b)). Therefore the result holds.
	
\end{proof}

\subsubsection*{Remark 2(b) :} 

\begin{proof}
	
	$$
	\begin{aligned}
		CATE(y_2, a)&=E_\lambda[y_{i2}|a,c_i=1]-E_\lambda[y_{i2}|a,c_i=0]\\
		&=\Phi(-a)-E_\lambda[\Phi(\lambda_i-a)]\\
	\end{aligned}
	$$
	
	Now, we compare $CATE(y_2, 0)$ and $CATE(y_2, 1)$. Firstly, we define $g(\lambda)= \Phi(\lambda)-\Phi(\lambda-1)+\Phi(-1)-\Phi(0)$. 
	
	Note that when $0<\lambda < 1$, $g(\lambda) > 0$. This can be seen from $g(0)=g(1)=0$, and it increases ($g'(\lambda)>0$) when $0<\lambda \leq 0.5$, and decrease (($g'(\lambda)<0$)) when $0.5 \leq \lambda<1$. 
	
	Therefore, 
	
	$$
	\begin{aligned}
		CATE(y_2,1)-CATE(y_2,0) = \mathbb{E}_\lambda [g(\lambda)] > 0.
	\end{aligned}
	$$
	
	Recall the treatment effect is negative (for example from the proof of Remark 1(b)). As a result,  $|CATE(y_2,1)|<|CATE(y_2,0)|$.
	
	%The results for $|CATE(y_2,-1)|<|CATE(y_2,0)|$ follow the same argument, with $g(\lambda)=\Phi(\lambda)-\Phi(\lambda+1)+\Phi(1)-\Phi(0)$.
	
	Next, we compare $CATE(y_2,1)$ and $CATE(y_2,-1)$. Define $h(\lambda)=\Phi(\lambda+1)-\Phi(\lambda-1)+\Phi(-1)-\Phi(1)$. Note that $\frac{d}{d\lambda}[\Phi(\lambda+1)-\Phi(\lambda-1)]=\phi(\lambda+1)-\phi(\lambda-1)<0$ for $\lambda>0$. Therefore, $h(\lambda)<0$ for $\lambda>0$. Then, 
	$$
	\begin{aligned}
		CATE(y_2,1)-CATE(y_2,-1) = \mathbb{E}_\lambda [h(\lambda)] < 0.
	\end{aligned}
	$$
	
	Because the treatment effect is negative, we conclude$|CATE(y_2,-1)|<|CATE(y_2,1)|$.
	
\end{proof}

\subsubsection*{Remark 2(c) :} 

Follows directly from Remarks 2(a) and 2(b).

\subsection[Simulation]{Simulation}\label{app:sim}

We conduct a simulation based on data generating processes based on our motivating theoretical model. Specifically, we generate datasets with $n \in \{100, 500, 1000\}$ observations under the following parametric and distributional assumptions:
\begin{itemize}
\item $c_i \in \{0, 1\}$ and $\Pr(c_i = 1) = \frac{1}{2}$
\item $a_i \in \{-1, 0, 1\}$ where $\Pr(a_i = -1) = \Pr(a_i = 0) = \Pr(a_i = 1) = \frac{1}{3}$
\item $\lambda \sim U[0, 1]$
\item $\varepsilon \sim \mathcal{N}(\mu, 1)$ for $\mu = \Phi^{-1}(0.1,0.2,...0.9)+0.25$, where $\Phi^{-1}(\cdot)$ is the inverse CDF of the standard normal distribution. We add $-E[c_i\lambda_i + a_i] = 0.25$ to generate simulated electorates with 10\%, 20\%,...,90\% incumbent support (in expectation). By varying $\mu$, therefore, we vary the \emph{ex-post} vote-share for the incumbent that the researcher would observe \emph{after} running the experiment. 
\end{itemize}

We generate outcome variables as in the main text:
\begin{align*}
y_{i1} &= -c_i\lambda_i + a_i + \varepsilon_i\\
y_{i2} &= I[y_{i1}\geq 0]
\end{align*}

Given each simulated data set, we estimate the following OLS model: 
\begin{align}\label{eq:estimator}
y_{il} &= \beta_0 + \beta_1 c_i + \beta_2 \lambda_i + \beta_3 a_i + \beta_4 c_i \lambda_i + \beta_5 c_i a_i + \epsilon_i,
\end{align}
for $l \in \{1, 2\}$ (both outcomes). We are interested in estimates of $\beta_4$ and $\beta_5$. $\beta_4$ is the estimator of the difference in CATEs  moving from the minimum to the maximum of corruption aversion, $\lambda$. $\beta_5$ is the estimator of the difference in CATEs for a one category shift in partisanship (e.g., opposition-aligned to non-aligned or non-aligned to aligned). \\

Figure \ref{fig:sim1} reports the average estimates of the differences in CATEs ($\widehat{\beta_4}$ and $\widehat{\beta_5}$ from \eqref{eq:estimator}). Figure \ref{fig:sim2} reports the power of the estimators $\beta_4$ and $\beta_5$ at the $\alpha = 0.05$ level. Since inferences about mechanisms are typically made on the basis of whether we detect HTEs (not the magnitude of the difference), Figure \ref{fig:sim2} corresponds to more typical use of HTEs. Note that for all sample sizes, in highly imbalanced electorates (i.e., average incumbent support of 10\% or 90\%), the $\beta_5$ is a better powered estimator than $\beta_4$.\\

\begin{figure}
\resizebox{\textwidth}{!}{\includegraphics{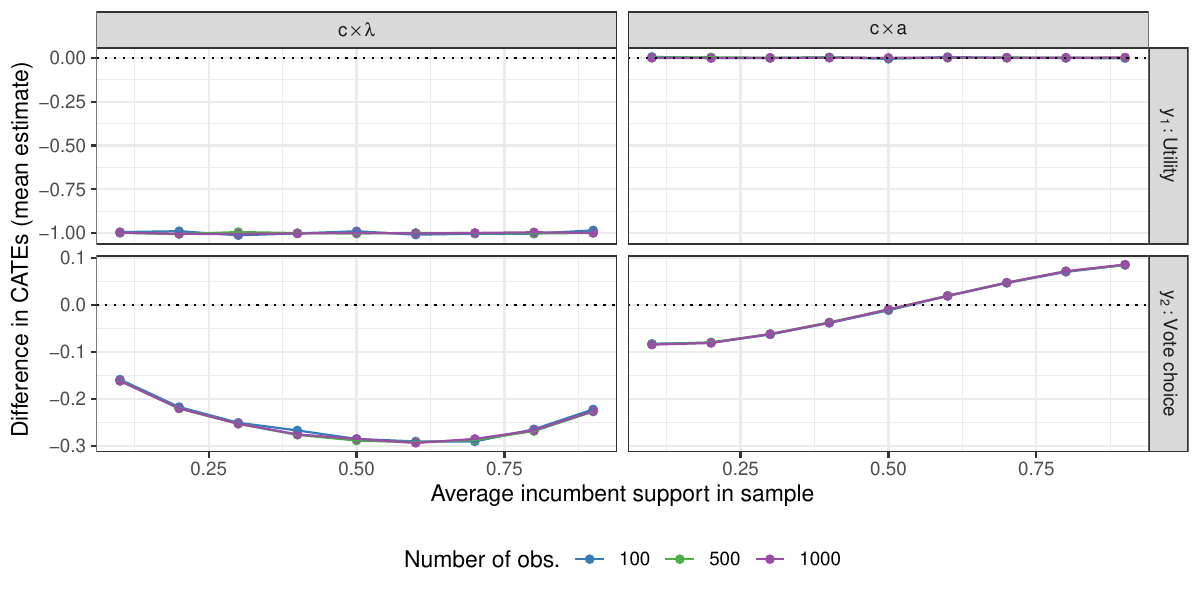}}
\caption{Difference in CATEs in each covariate ($\lambda_i$ and $a_i$) for outcomes $y_{i1}$ and $y_{i2}$. Each point reports the average over 5,000 simulated datasets.}\label{fig:sim1}
\end{figure}

\begin{figure}
\resizebox{\textwidth}{!}{\includegraphics{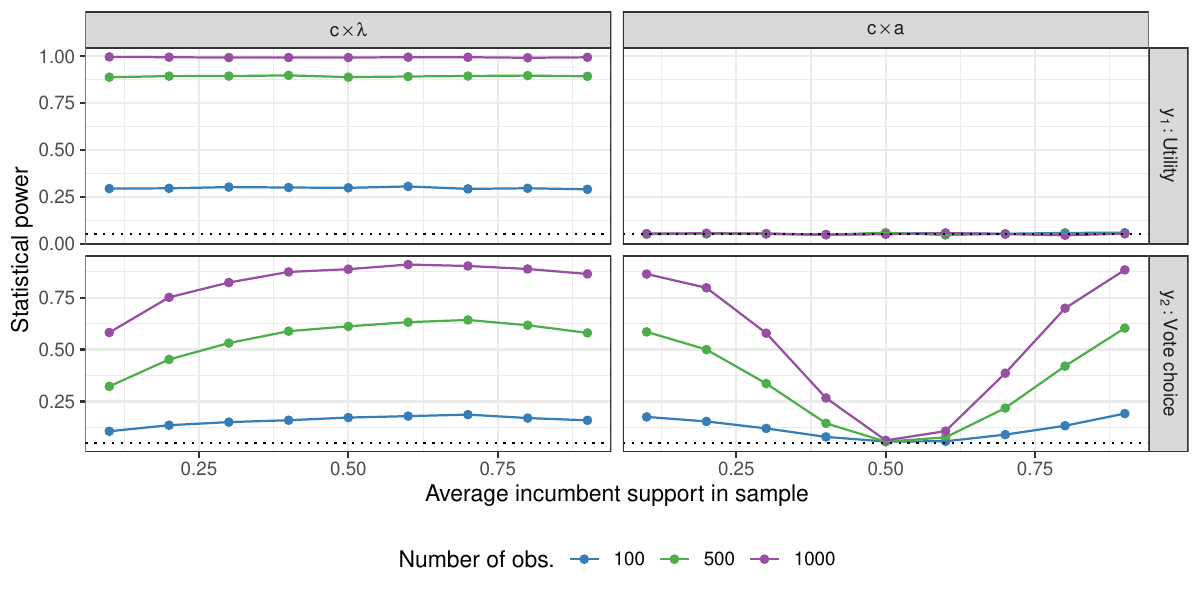}}
\caption{Power of the difference in CATEs in each covariate ($\lambda_i$ and $a_i$) for outcomes $y_{i1}$ and $y_{i2}$. Each point reports the average over 5,000 simulated datasets.}\label{fig:sim2}
\end{figure}

Given the discrete distribution of $a_i$, researchers may opt for a more flexible regression specification than \eqref{eq:estimator}. To this end, we also consider the following model that compares effects of corruption revelation among independent/neutral ($a_i = 0$) and aligned ($a_i = 1$) partisans relative to its effects among unaligned ($a_i = -1$) partisans:
\begin{align}\label{eq:estimator_factor}
y_{il} = \beta_0 & + \beta_1 c_i + \beta_2 \lambda_i + \beta_3 \mathbb{I}(a_i=0) + \beta_4 \mathbb{I}(a_i = 1) + \beta_5 c_i \lambda_i + \beta_6 c_i \mathbb{I}(a_i = 0) + \beta_7 c_i \mathbb{I}(a_i = 1) + \epsilon_i. \end{align}
Figure \ref{fig:sim1factor} plots average estimates of the differences in CATEs ($\widehat{\beta_5}$, $\widehat{\beta_6}$, and $\widehat{\beta_7}$ in \eqref{eq:estimator_factor}) while Figure \ref{fig:sim2factor} plots the power of these estimators at the $\alpha = 0.05$. We see that particularly in highly imbalanced electorates, researchers are as likely to detect differences in the CATEs between unaligned (base category, $a_i = -1$) and aligned ($a_i = 1$) voters as they are to detect differences in CATEs in corruption aversion ($\lambda_i$). 

\begin{figure}
\resizebox{\textwidth}{!}{\includegraphics{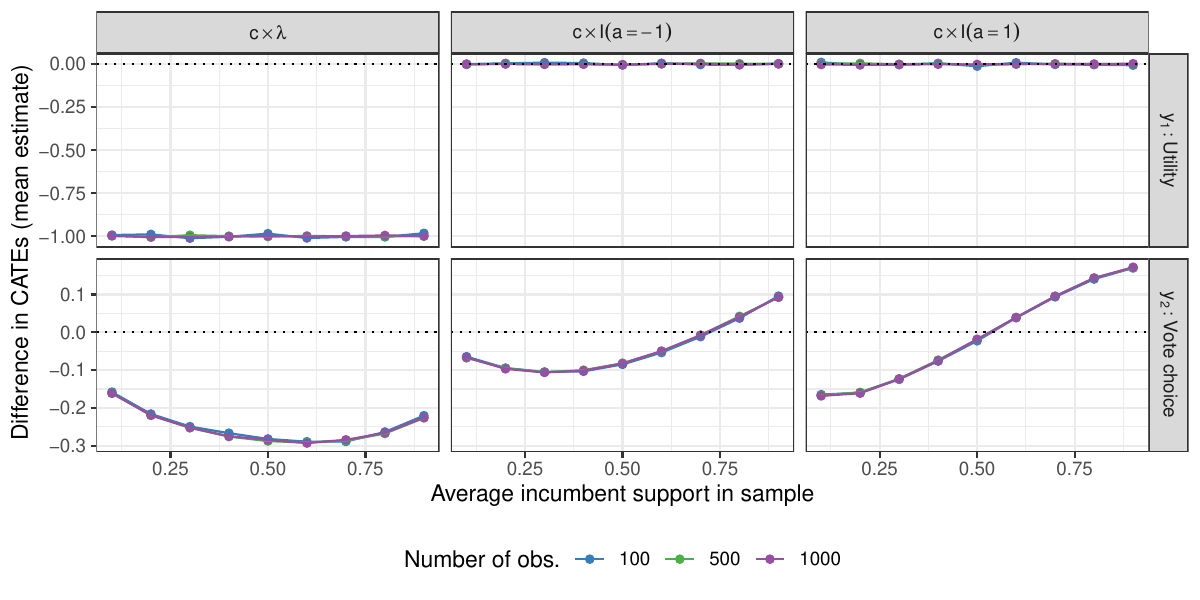}}
\caption{Difference in CATEs in each covariate ($\lambda_i$ and $a_i$) for outcomes $y_{i1}$ and $y_{i2}$. Each point reports the average over 5,000 simulated datasets.}\label{fig:sim1factor}
\end{figure}

\begin{figure}
\resizebox{\textwidth}{!}{\includegraphics{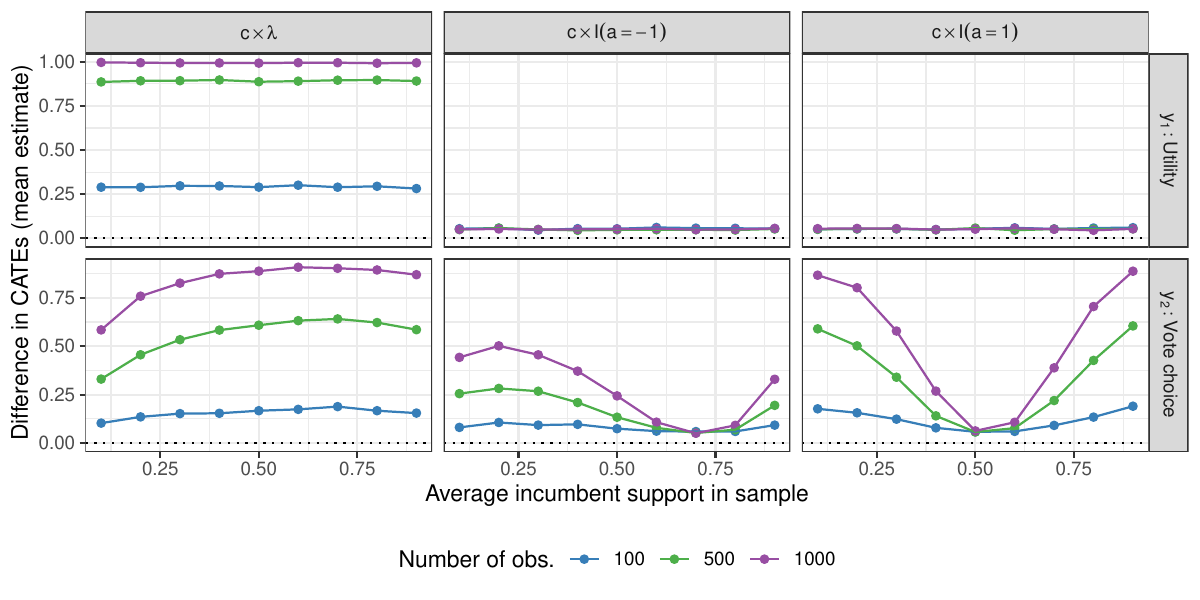}}
\caption{Power of the difference in CATEs in each covariate ($\lambda_i$ and $a_i$) for outcomes $y_{i1}$ and $y_{i2}$. Each point reports the average over 5,000 simulated datasets.}\label{fig:sim2factor}
\end{figure}

\section[  Generalized Framework with Multiple Mechanisms]{Generalized Framework with Multiple Mechanisms}\label{app:multi}

Suppose that there exist $J\geq1$ mediators (or mechanisms), indexed by $M_1,...,M_J$. Given two treatment values, $z, z' \in Z$, the total effect of $Z$ on $Y$ is:
\begin{align}
	TE^Y(z, z'; X) &= Y(z, M_1(z; X), ..., M_J(z; X); X)- Y(z', M_1(z'; X),..., M_J(z'; X); X)
\end{align}

Our notation varies slightly from conventional presentations of mediation that only consider two mechanisms (mediators) \citep{imai2010general}. To this end, $j-$ and $j+$ denote index $h \in J$ such that $h<j$ and  $h>j$ respectively. Treatment effects may consist of direct ($DE$) and indirect ($IE_j$) effects, as follows:\footnote{The intuition for our notation is as follows: define $IE_0(z, z'; X)= DE(z, z'; X)$. This implies that the first term of $IE_j(z, z'; X)$ and the second term of  $IE_{j-1}(z, z';X)$ cancel out.}
\begin{align}
	DE^Y(z, z'; X) = Y(z, M_1(z; X), ..., M_J(z; X); X)- Y(z', M_1(z; X),...,M_J(z; X); X)
\end{align}

\begin{equation}
	\begin{aligned}
		IE^Y_j(z, z'; X) = Y&(z', M_{j-}(z'; X) ,M_j(z; X),M_{j+}(z; X) ; X)- \\ 
		&Y(z', M_{j-}(z'; X) ,M_j(z'; X),M_{j+}(z; X) ; X)
	\end{aligned}
\end{equation}

% IE_j(z; X) &= Y(z, M_j(z, X), M_{\neg j}(z, X), X)- Y(t, M_j(z', X), M_{\neg j}(z, X), X)

The direct effect, $DE^Y(z, z';X)$ represents the direct effect of $Z$ on $Y$ holding mediators at potential outcomes  $M_j(z; X)$. It is not necessary to believe that treatments produce unmediated (direct) effects on outcomes to use this framework. We allow for direct effects in the interest of generality. The indirect effects, measures the effect on the outcome that operates by changing the potential outcome of the mediator. As is standard, we can then re-write the total effect as follows:

\begin{align}\label{eq:te1}
	TE^Y(z, z';X) &= DE^Y(z, z';X)+ \sum_{j= 1}^J IE^Y_j(z,z'; X), 
\end{align}
which is defined at the unit, or individual level. Although there exists multiple ways to decompose total effect and this decomposition holds mechanisms to be independent, all results in the paper hold if other mediators (other than mediator $j$ of interest) are arbitrarily correlated.

 If we evaluate expectations over $X$, we obtain:
\begin{align}
	ATE^Y(z, z') &= E_{X}[Y(z,M_j(z;X);X)- Y(z',M_j(z';X);X)] \\
	&= E_{X}[DE^Y(z, z'; X) + \sum_{j = 1}^J IE^Y_j(z,z'; X)] \label{equ:ade1}\\
	&=  ADE^Y(z, z')+\sum_{j = 1}^J AIE^Y_j(z,z')\label{equ:aie1}
\end{align}

Because there are multiple mediators, we generalize Assumption \ref{cond:irre2} from the main text while maintaining Assumption \ref{cond:irre1}

\begin{ass}[Exclusion II]\label{cond:irrea2}
	Given $z,z' \in Z$ and $x,x' \in X_k$, $X_k$ is excluded to the indirect effect of any other mechanism, $j' \neq j$, $AIE^Y_{j'}$: $AIE^Y_{j'}(X_k=x)=AIE^Y_{j'}(X_k=x')$.
\end{ass}

The only difference between Assumption \ref{cond:irre2} and Assumption \ref{cond:irrea2} is that Assumption \ref{cond:irrea2} is a generalization to Mechanisms $j \in 3,...,J$. 

\section[  Comparison to Mediation Analysis]{Comparison to Mediation Analysis}\label{app:mediation}

In this section, we compare our framework connecting HTEs and mechanisms to mediation analysis. It is important to note that the two frameworks are built on different principles and objects. The main purpose of mediation analysis is to identify and estimate various average causal mediation effects (ACMEs). Identification of these effects relies on the assumption of sequential ignorability \citep{imai2013identification}. On the other hand, our framework aims to infer the activation of a mechanism (for at least one unit in the sample) by using heterogeneous treatment effects, which instead, relies on exclusion assumptions that we propose (Assumptions \ref{cond:irre1}-\ref{cond:irre2}). The following DAGs facilitate our discussion of the differences in these approaches. \\

\begin{figure}[!h]
	\begin{center}
		\begin{tikzpicture}
		\node at (2, 2) {\textsc{Mediation framework}};
			\node at (0,0) (z) {$Z$};
			\node at (2,0) (m) {$M_1$};
			\node at (2,0.8) (m') {$M_2$};
			\node at (4,0) (y) {$Y$};
			\node at (3,-1) (u) {$U$};
			\draw[->] (z) -- (m);
			\draw[->] (m) -- (y);
			\draw[->,red] (u) -- (m);
			\draw[->,red] (u) -- (y);
			\draw[->] (z) -- (m');
			\draw[->] (m') -- (y);
			\draw[->,dashed,red] (z) -- (u);
			\draw[->]
			(z) edge [bend left=80] (y);         
			
	\node at (8, 2) {\textsc{Our framework}};
			\node at (6,0) (z) {$Z$};
			\node at (8,0) (m) {$M_1$};
			\node at (8,0.8) (m') {$M_2$};
			\node at (10,0) (y) {$Y$};
			\node at (9,-1) (u) {$U$};
			\node at (7,-1) (x) {$X$}; 
			\node at (6.6,0.1) (int1) {}; 
			\node at (7.4,0.6) (int2) {}; 
			\node at (7,1.3) (int3) {}; 
			\draw[->] (z) -- (m);
			\draw[->] (m) -- (y);
			\draw[->] (u) -- (m);
			\draw[->] (u) -- (y);
			\draw[->] (z) -- (m');
			\draw[->] (m') -- (y);
			\draw[->] (x) -- (int1);
			\draw[->,dashed,red] (x) -- (int2);
			\draw[->,dashed,red] (x) -- (int3);
			\draw[->]
			(z) edge [bend left=80] (y);    
			
		\end{tikzpicture}
	\end{center}
\caption{DAGs when mechanisms $M_1$ and $M_2$ are independent. Note that the red arrows are ruled out by assumption of each respective framework. The left DAG, representing the assumptions of the causal mediation framework, highlights that all variables $U$ must be in the adjustment set and also cannot be affected by the treatment $Z$. The right DAG, representing our framework, emphasizes that if covariate $X$ seeks to measure the activation of $M_1$, it must not moderate other channels.} \label{fig:seq1}
\end{figure}

Consider first the case with multiple independent causal mechanisms in Figure \ref{fig:seq1}. In the left DAG, treatment $Z$ indirectly affects outcome $Y$ through two channels $M_1$ and $M_2$, and may also directly affect $Y$. To non-parametrically identify average indirect effect mediated by $M_1$, the key part of the sequential ignorability is $Y_i(z',m_1,M_{i2}(z')) \perp M_{i1}|Z_i=z$. The assumption is challenging to interpret and cannot be guaranteed by most experimental designs because it involves cross-index independence, from $z$ to $z'$. Graphically, sequential ignorability requires that all variables $U$ should be observed and included in the adjustment set. Another important implication of the assumption is that $U$ cannot be affected by the treatment $Z$, i.e., the dashed line is not allowed. When these assumptions hold, mediator $M_1$ must be measured in order to estimate the ACME.\\

In the right DAG, treatment $Z$ again affects outcome $Y$ through two channels $M_1$ and $M_2$, and may also directly affect $Y$. Mechanism $M_1$ is activated if its average indirect effect is non-zero for some unit. The existence of HTEs provides a sufficient condition for this activation under both exclusion assumptions. In our framework, variables $U$ may or may not be measured or included in the adjustment set. Further, $U$ can be a child of treatment, $Z$ (though this introduces a third mechanism). For HTEs to (ever) be informative of mechanism activation, we need to observe another pre-treatment variable $X$. It is assumed not to moderate (average) direct effect and (average) indirect effect mediated by $M_2$. That is, two dashed lines are excluded in the right DAG in Figure \ref{fig:seq1}. From these DAGs, it is clear that there is no logical ordering of the strength of the two types of assumptions. \\

There are many other differences between the two methods. For example, in the mediation analysis, mediators must be measurable and measured while these measurements are not required by our framework. However, in our framework, researchers must have a measured candidate MDV $X$ that is believed to satisfy Assumptions \ref{cond:irre1}-\ref{cond:irre2}. Also, when using HTEs to detect mechanisms, researchers need to pay more attention to whether $Y$ is directly affected outcome. Even though we have emphasized their differences, two frameworks also have shared features. For example, both require that the total causal effect of $Z$ on $Y$ is identified. This is explicitly assumed by sequential ignorability $\{Y_i(z,m_{i1},m_{i2}), M_{i1} ,M_{i2} \} \perp  Z_i$ and implicitly assumed in our framework.

\subsection{Related Mechanisms and Correlated Mediators}

Because extension to related mechanisms (correlated mediators) is not our main focus, we only make some brief comments. Consider two different correlation structures. In the left DAG of Figure \ref{fig:seq2}, mediator $M_2$ directly affects $M_1$. As mentioned in the \citet{imai2013identification}, two assumptions are required to identify the ACME with respect to $M_1$. The first one is the modified sequential ignorability assumption. Unfortunately, with causally dependent multiple mediators, an assumption of no treatment-mediator interaction effects is also required. For the HTE-mechanism framework, we can simply treat correlated mechanisms as one (molar) mechanism. Then, as long as exclusion assumptions hold for the average direct effect and other indirect effects, our results in the main text still hold. One caveat is that if the mechanism represented by $M_1$ is inactive, for example, because the dashed line in the figure disappears, then the HTE-mechanism framework may yield misleading inferences about the influence of mechanism 1.  \\

\begin{figure}[!h]
	\begin{center}
		\begin{tikzpicture}
			\node at (0,0) (z) {$Z$};
			\node at (2,0) (m) {$M_1$};
			\node at (2,1) (m') {$M_2$};
			\node at (4,0) (y) {$Y$};
		%	\node at (3,-1) (u) {$U$};
			\draw[->] (z) -- (m);
			\draw[->, dashed] (m) -- (y);
		%	\draw[->] (u) -- (m);
		%	\draw[->] (u) -- (y);
			\draw[->] (z) -- (m');
			\draw[->] (m') -- (y);
		%	\draw[->,dashed,red] (z) -- (u);
			\draw[->,blue] (m') -- (m);
			\draw[->]
			(z) edge [bend left=80] (y);

			\node at (6,0) (z) {$Z$};
			\node at (8,0) (m) {$M_1$};
			\node at (8,1) (m') {$M_2$};
			\node at (10,0) (y) {$Y$};
		%	\node at (9,-1) (u) {$U$};
			\node at (7,-1) (x) {$X$}; 
			\node at (6.6,0.1) (int1) {}; 
			\node at (7.4,0.6) (int2) {}; 
			\node at (7,1.3) (int3) {}; 
			\draw[->] (z) -- (m);
			\draw[->] (m) -- (y);
		%	\draw[->] (u) -- (m);
		%	\draw[->] (u) -- (y);
			\draw[->] (z) -- (m');
			\draw[->] (m') -- (y);
			\draw[->,blue] (x) -- (int1);
			\draw[->,blue] (x) -- (m');
		%	\draw[->,dashed,red] (x) -- (int2);
		%	\draw[->,dashed,red] (x) -- (int3);
			\draw[->]
			(z) edge [bend left=80] (y);    
			
		\end{tikzpicture}
	\end{center}
	\caption{Two DAGs with correlated mediators. The blue arrows represent correlation structures that can be accommodated using HTE analysis within our framework.} \label{fig:seq2}
\end{figure}

Multiple mechanisms can also be correlated due to other common covariates. This correlated structure can be easily accommodated to the HTE-mechanism framework. For example, in the right DAG of Figure \ref{fig:seq2}, variable $X$ affects two indirect channels with respect to $M_1$ and $M_2$. However, $X$ does not moderate the $M_2$ channel and the direct channel, and thus exclusion assumptions hold. In this case, our results can be directly applied without any modification. Note that in the mediation analysis, this related structure is still classified as having independent causal mechanisms.

\section[  Proofs of Propositions 1-4]{Proofs of Propositions}

We prove propositions under general multi-mechanism cases, as stated in the \ref{app:multi}. 

\subsection{Proof of Proposition 1}
\begin{prop}\label{propid_app}
Assumptions \ref{cond:irre1}-\ref{cond:irre2} are sufficient and generically necessary for $CATE(X_k=x)-CATE(X_k = x')= AIE^Y_1(z,z';X_k = x) - AIE^Y_1(z,z';X_k = x')$.
\end{prop}

We prove a stronger version of Proposition \ref{propid_app} for any non-zero $L(Y)$ where $L$ is a non-zero linear transformation. By non-zero linear transformation, we mean that there exists a non-zero constant matrix $A$ such that $L(Y)=AY$.

\begin{proof}
By definition of $CATE$:

\begin{align}
	CATE^{L(Y)}(X_k = x) &= E_{X_{\neg k}} [L(Y)|Z = z, X_k= x]-E_{X_{\neg k}}[L(Y)|Z = z', X_k= x] \\
	&=E_{X_{\neg k}} [L(DE^Y(z,z';X_k=x) + \sum_{j=1}^J IE^Y_j(z,z';X_k=x))]\label{eq:1a}\\
	&= L \{E_{X_{\neg k}} [DE^Y(z,z';X_k=x) + \sum_{j=1}^J IE^Y_j(z,z';X_k=x)]\} \label{eq:1b} \\
	&= L [ADE^Y(z,z';X_k=x) + \sum_{j=1}^J AIE^Y_j(z,z';X_k=x)]
\end{align}

Equation \eqref{eq:1a} follows from the linearity of expectations and the decomposition of the total effect in \ref{eq:te1}. Equation \eqref{eq:1b} is guaranteed by the linearity of $L$. \\

Therefore, 
\begin{align}
&	CATE^{L(Y)}(X_k = x) - CATE^{L(Y)}(X_k = x') \\
=&L[ADE^Y(z,z';X_k=x)-ADE^Y(z,z';X_k=x')] \label{eq:idi1}\\
+&L[\sum_{j=2}^J AIE^Y_j(z,z';X_k=x)-AIE^Y_j(z,z';X_k=x')]\label{eq:idi2}\\
 +& L[AIE^Y_1(z,z';X_k=x)-AIE^Y_1(z,z';X_k=x')]
\end{align}

For sufficiency, under exclusion assumptions \ref{cond:irre1} and \ref{cond:irre2}, equation \eqref{eq:idi1} and \eqref{eq:idi2} are zero.\\

For generically necessity,   let $\Delta ADE^Y$ denote  $ADE^Y(z,z';X_k=x)-ADE^Y(z,z';X_k=x')$, and $\Delta AIE^Y$ denote $\sum_{j=2}^J AIE^Y_j(z,z';X_k=x)-AIE^Y_j(z,z';X_k=x')$. Then, $CATE^{L(Y)}(X_k = x) - CATE^{L(Y)}(X_k = x')=L[AIE^Y_1(z,z';X_k=x)-AIE^Y_1(z,z';X_k=x')]$ implies $L[\Delta ADE^Y + \Delta AIE^Y]=0$.  That is, (1)$\Delta ADE^Y =0$ and $\Delta AIE^Y=0$, and (2) $\Delta ADE^Y \neq 0$ or/and $\Delta AIE^Y\neq0$, and $\Delta ADE^Y = - \Delta AIE^Y$. \\

Case (1) is equivalent to assumptions \ref{cond:irre1}-\ref{cond:irre2}. This can be realized if covariate $X_k$ does not `modify' the average direct and indirect effects. \\

For case (2), covariate $X_k$ does modify the average direct or/and indirect effects. Generically, it can take any real numbers. However, the probability that two real numbers exactly offset, $\Delta ADE^Y = - \Delta AIE^Y$ is 0. \\

Therefore, we conclude that generically, case (1), i.e., assumptions \ref{cond:irre1}-\ref{cond:irre2}, hold.

\end{proof}

\subsection{Proof of Proposition 2}

\begin{prop}\label{prop1_app}Suppose Assumptions \ref{cond:irre1}-\ref{cond:irre2} hold with respect to $X_k$ for outcome $Y$. If HTEs exist with respect to $X_k$, then $X_k \in \textbf{X}^{MDV}$ for mechanism $M_j$.
\end{prop}

We prove a stronger version of Proposition \ref{prop1_app} for any non-zero $L(Y)$ where $L$ is a non-zero linear transformation. By non-zero linear transformation, we mean that there exists a non-zero constant matrix $A$ such that $L(Y)=AY$.

\begin{proof}
Follow the same proof in the Proposition \ref{propid_app}, we have $CATE^{L(Y)}(X_k = x)= L [ADE^Y(z,z';X_k=x) + \sum_{j=1}^J AIE^Y_j(z,z';X_k=x)]$.

Then, under exclusion assumptions \ref{cond:irre1} and \ref{cond:irre2}, we can express:

%\begin{align}\label{equ:hte}
%CATE^{L(Y)}(X_j = x) - CATE^{L(Y)}(X_j = x')=L\{E_{X_{\neg k}} [ %IE_j(z,z';X_k=x)-IE_j(z,z';X_k=x']\}
%\end{align}

\begin{align}\label{equ:hte}
	CATE^{L(Y)}(X_k = x) - CATE^{L(Y)}(X_j = x')=L[ AIE^Y_j(z,z';X_k=x)-AIE^Y_j(z,z';X_k=x')]
\end{align}

HTEs  exist with respect to $X_k$ if \eqref{equ:hte} is non-zero. In this case, then $X_k \in \textbf{X}_j^{MDV}$ by the definition of MDV.

\end{proof}

\subsection{Proof of Proposition 3}
  
\begin{prop}\label{prop2_app}Suppose Assumptions \ref{cond:irre1}-\ref{cond:irre2} hold with respect to $X_k$ for outcome $Y$. If no HTEs exist with respect to $X_k$, at least one of the following must be true: If no HTEs exist with respect to $X_k$, at least one of the following must be true:
	\begin{enumerate}
		\item $X_k \notin \textbf{X}^{MDV}$ for mechanism $j$.
		\item  No MDV exists.
	\end{enumerate}
\end{prop}

We prove a stronger version of Proposition \ref{prop2_app} for any non-zero $L(Y)$ where $L$ is a non-zero linear transformation. By non-zero linear transformation, we mean that there exists a non-zero constant matrix $A$ such that $L(Y)=AY$.
 
 \begin{proof}
 Proof by contrapositive. Suppose not, which means  $\textbf{X}_j^{MDV}$ is non-empty and for some $x,x' \in \mathbb{R}$, $AIE^Y_j(X_k = x) \neq AIE^Y_j(X_k = x')$. Then: 

We then can reconstruct $CATE^{L(Y)}(X_k = x) \neq CATE^{L(Y)}(X_k = x')$.

\begin{align}
	L\{ AIE^Y_j(z,z';X_k=x) \} &\neq 	L\{ A IE^Y_j(z,z';X_k=x') \}\\
	 L\{ADE^Y(X_k=x)+ AIE^Y_j(X_k=x) +AIE^Y_{i \neq j}(X_k=x) \} & \neq  L\{ADE^Y(X_k=x')+ \sum_{j=1}^J AIE^Y_j(z,z';X_k=x) \} \label{eq2a}\\
	CATE^{L(Y)}(X_j = x) &\neq CATE^{L(Y)}(X_j = x')
\end{align}

Equation \eqref{eq2a} follows because $ADE^Y(X_k=x)=ADE^Y(X_k=x')$ and $AIE^Y_{i \neq j}(X_k=x)=AIE^Y_{i \neq j}(X_k=x')$ under exclusion assumptions \ref{cond:irre1} and \ref{cond:irre2}. We find HTEs with respect to $X_k$. \\
 
We have thus shown that if two conditions do not hold, then HTEs exist for $X_k=x$ and $X_k=x'$. By contrapositive, we prove that if no HTEs exist with respect to $X_k$, at least one of the two conditions must be true.

  \end{proof}

\subsection{Proof of Proposition 4}

\begin{prop}\label{prop3_app}
Suppose Assumptions \ref{cond:irre1}-\ref{cond:irre2} hold with respect to $X_k$ for outcome $Y$. Let observed outcome $h(Y)$ be a transformed outcome relative to $Y$.
If HTEs exist with respect to $X_k$ for outcome $h(Y)$, then $X_k \in \textbf{X}^{R}.$
\end{prop}

\begin{proof}
Proof by contrapositive. By definition of $\textbf{X}^{R}$, we know $X \notin \textbf{X}^{R}$ implies that $X$ must be independent of $Y$.  Let  $\mathbb{P}(\tilde{y}|Z, X_k)$  be the conditional distribution of $h(Y)$. Suppose $X_k \notin \textbf{X}^{R}$, then:

\begin{align}
	CATE(X_k=x) &=\int \tilde{y} d[ \mathbb{P}(\tilde{y}|Z=z, X_k=x)-\mathbb{P}(\tilde{y}|Z=z',X_k=x)]\\
	&=\int \tilde{y} d[ \mathbb{P}(\tilde{y}|Z=z)-\mathbb{P}(\tilde{y}|Z=z')]\label{eq3a}\\
	&=\int \tilde{y} d[ \mathbb{P}(\tilde{y}|Z=z, X_k=x')-\mathbb{P}(\tilde{y}|Z=z', X_k=x')]\label{eq3b}\\
	&=CATE(X_k=x')
\end{align}

Equations \eqref{eq3a} and \eqref{eq3b} follow from the fact that $X_k$ is independent of $Y$ if $X_k \notin \textbf{X}^R$. \\

Therefore, equivalently, we have shown if HTEs exist with respect to $X_k$, then $X_k \in \textbf{X}^{R}$ by contrapositive.

\end{proof}

The preceding proof does not explicitly invoke the definition of \(h(\cdot)\) or rely on the assumptions concerning \(Y\). A more direct -- though considerably more tedious -- proof can be carried out by introducing \(\delta\) and performing the required integrations.\\

A more important question is whether \(X_k \in \mathbf{X}^R \setminus \mathbf{X}^{MDV}\), meaning that \(X_k\) belongs to \(\mathbf{X}^R\) but not to \(\mathbf{X}^{MDV}\). We can demonstrate that this is possible with a simple example. Suppose that we are interested in how a mobilization treatment, $Z_i \in \{0,1\}$, affects citizens' decision to vote. Consider two covariates that predict turnout: $X_1 \sim \mathcal{N}(0,1)$ and $X_2 \in \{0,1\}$. We will further assume that $X_1$ is a MDV for the unique mechanism, such that:
\begin{align*}
	M(Z; \textbf{X}) &= (1+Z)X_{1}
\end{align*}
Potential voters' utility from voting is given by:
\begin{align*}
	U(Z, X) &= M(Z; X_{1}) + X_{2} = (1+Z)X_{1} + X_{2}
\end{align*}
Our observed behavioral outcome---turnout---is a non-linear function of voter utility as follows:
\begin{align*}
	h(U(Z, X)) & = \begin{cases}
		1 & \text{ if }(1+Z)X_{1} + X_{2} \geq 0\\
		0 & \text{ else }
	\end{cases}
\end{align*}

In this case, Assumptions \ref{cond:irre1}-\ref{cond:irre2} hold by construction since there is only one mechanism. Now, suppose that a researcher mistakenly thought that $X_{2}$ was a MDV for a mechanism (either $M$ or an inert mechanism). Since we have constructed the data generating process, we know that it is not: $X_{2} \in \textbf{X}^R \setminus \textbf{X}^{MDV}$. Evaluating the CATE of $Z$ on turnout when $X_2 = 1$, we have:
\begin{align*}
	CATE(X_2 = 1) &= E[h(U(Z = 1, X))-h(U(Z=0, X))|X_2 = 1]\\
	&= \Pr(2 X_1 + 1 > 0) - \Pr(X_1 + 1 > 0)\\
	&= \Phi(-1)-\Phi(-\frac{1}{2})\\
	&\approx -0.15
\end{align*}
Note that $\Phi(\cdot)$ is the cdf of the standard normal distribution, which we invoke because $X_1 \sim \mathcal{N}(0,1)$. Using the same approach it is straightforward to see that $CATE(X_2 = 0) = \Phi(0)-\Phi(0) = 0$. So it is clear that we have HTEs in $X_2$ because $CATE(X_2 = 1) \neq CATE(X_2 = 0)$. Remember that, by construction, $X_2$ is \emph{not} a MDV, though it is relevant. This numerical example is analogous to the issue that arises with valence in the analysis of our motivating example. \\

For additional intuition about how non-linear transformations of $Y$ affect HTE,  consider the following decomposition: Given pre-treatment variables $\{X_1, X_2,...,X_n\} = \textbf{X}$. Let variables $\{X_{1},...,X_m \} \subset \textbf{X}$ are MDVs for mechanism j  (denote them as the set $\textbf{X}_j^{MDV}$ and the remaining as the set $\textbf{X}_j^{non-MDV}$). 
	
	\begin{equation}\label{equ:additive}
		Y=g_1( \textbf{X}_j^{non-MDV}, \textbf{X}_j^{MDV})+g_2( \textbf{X}_j^{MDV},Z)
	\end{equation} where $g_1(\cdot)$ is any function while  $g_2(\cdot)$ is non-additively separable. A function, $F(X_1,X_2)$, will be called additively separable if it can written as $f_1(X_1) + f_2(X_2)$ for some functions $f_1(X_1)$ and $f_2(X_2)$. 
	
Note  that:
\begin{enumerate}

\item \eqref{equ:additive} should be understood as $$Y=g_1(X_1, X_2,...,X_n)+g_2(X_1,...,X_m,Z).$$

\item The non-additively separable function $g_2( \textbf{X}_j^{MDV},Z)$ can take the form $g_3(Z)+g_4 (\textbf{X}_j^{MDV},Z)$ for some function $g_3(\cdot)$ and non-additively separable function $g_4(\cdot)$.
\end{enumerate}

For any non-zero linear transformation of $Y$, $h(Y)$, calculation of conditional expectations yields:
\begin{align}
	CATE(X=x)-CATE(X=x')=E[h_2( g_2(\textbf{X}_j^{MDV},Z))|X=x]-E [h_2(g_2( \textbf{X}_j^{MDV},Z'))|X=x'] \label{eq3c}
\end{align}
Equation \eqref{eq3c} is only a function of $\textbf{X}_j^{MDV}$ because $h_1(\cdot)$ cancels out.\\

However, for nonlinear transformations $h(Y)$, we cannot cancel $g_1(\cdot)$ in the absence of additional assumptions restricting the functional form of $h(Y)$.

\subsection{Extension of Proposition 5}

%\begin{proof}
%Recall, 
%$\delta(X_i=x)=\sum_{i=1}^k y_i [p_i(x;t)-p_i(x;t')]$ and $\delta(X_i=x')=\sum_{i=1}^k y_i [p_i(x';t)-p_i(x';t')]$.

%We regard $p_i(x;t)$ is a random variable. According to the assumption, $\exists$ a nonempty set $I \subseteq \{1,2,..,k\}$, $p_i(x;t), \forall i \in I$ takes value in the set $P_i$ that itself has non-zero Lebesgue measure; for example, $P$ can be a interval $[a,b]$ that $0\leq a < b \leq 1$. Therefore, with probability zero, $p_i(x;t)=p_i(x;t')=p_i(x';t)=p_i(x';t') \; \forall i \in I$. This implies with probability zero that $\delta(X_i=x)=\delta(X_i=x')=0$.

%Now, WLOG, we can focus on non-zero $\delta(X_i=x)= \sum_{i \in I} y_i [p_i(x;t)-p_i(x;t')]$ and non-zero $\delta(X_i=x')=  \sum_{i \in I} y_i [p_i(x';t)-p_i(x';t')]$. We then apply the assumption one more time, it has probability zero that $\sum_{i \in I} y_i [p_i(x;t)-p_i(x;t')]= \sum_{i \in I} y_i [p_i(x';t)-p_i(x';t')]$ for given $y_i$. This is particularly true if we think $y_i$ is just the numerical label (has no substantive meaning) that researcher sets arbitrarily.
%\end{proof}

%\section[  Complementary Propostion 4]{Complementary Propostions 4}

\begin{prop}\label{prop4_app}
Suppose Assumptions \ref{cond:irre1}-\ref{cond:irre2} hold with respect to $X_k$ for outcome $Y$. Let observed outcome $h(Y)$ be a transformed outcome relative to $Y$. If HTEs do not exist with respect to $X_k$ for outcome $h(Y)$, then $X_k \in \textbf{X}.$
\end{prop}

\begin{proof}
	The result follows simply because if $X_k \notin \textbf{X}$, then $X_k = \emptyset$.
\end{proof}

Proposition \ref{prop4_app} indicates that if there exist no HTEs for the transfomed outcome, $X_k$ can be any relevant or non-relevant covariate. The proof again does not explicitly invoke the definition of \(h(\cdot)\) or rely on the assumptions concerning \(Y\). Simple examples can show that $X_k$ can belong to $\textbf{X}^{MDV}$ and $\textbf{X}^{non-MDV}$. Now we provide a stronger version of Proposition \ref{prop4_app} by imposing assumptions about the outcome, $Y$, and the form of the non-linear transformation $h(Y)$. These assumptions permit additional learning from the lack of HTEs in this case.\\

In practice, most transformed outcomes are discrete variables, such as voting behavior, survey responses, or choices. Let us consider the following non-linear transformation of the directed affected outcome $Y$:

\begin{equation}\label{equ:disy}
	h(Y) = 
	\begin{cases}
		y_1 & Y \in (-\infty,c_1]\\
		y_2& Y \in (c_1,c_2]\\
		...\\
		y_q & Y \in (c_{q-1},\infty)\\
	\end{cases}
\end{equation}

Here, will assume $y_i \in \mathbb R$ in \eqref{equ:disy} has no substantive interpretation. In practice, values of $y_i$ are typically normalizations, that are
arbitrarily determined by the researcher. As such, the value is independent of model parameters. \\

To simplify some notation, we define:
\begin{align}
p_i(x;z) &\equiv \Pr [y\in(c_{i-1},c_i]|X=x,Z=z]\\
p_i(x; z, z') &\equiv Pr[y\in(c_{i-1},c_i]|X=x, Z=z]-\Pr [y\in(c_{i-1},c_i]|X=x,Z=z'].
\end{align}

Note that in the interest of parsimony, we omit $M$ in the above equations even though $Y$ is defined as a function of $Z, X,$ and $M$. We maintain $Z$ and $X$ because in order to calculate CATEs, we need at least two possible values of the treatment $Z$ and two distinct values of the covariate $X_k$. We define a covariate as \emph{effective} as follows:

\begin{defn}
 $X_k \in \textbf{X}$ is effective if $\exists i \in \{1,2,...,q\}$ and $x,x' \in X_k$ such that $p_i(x;z,z') \neq p_i(x';z,z')$. 
\end{defn}

Effectiveness means that as $X_i$ changes, it can induce a different probability of $h(Y)=y_i$. It should be clear that if $X_k$ is effective, then it must the case that $X_k \in \textbf{X}^{R}$. In general, if $X_k$ is not effective, then $X_k \notin \textbf{X}^{R}$. %Only under very specific circumstance, $X_k \in \textbf{X}^{R}$.
 
 \begin{prop}\label{propa4}
 	Suppose that observed outcome $h(Y)$ is a discrete non-linear mapping of directly-affected outcome $Y$ in equation \eqref{equ:disy} and Assumptions 1 and 2 hold. Assume further that $Y$ has an absolutely continuous distribution. If HTEs do not exist with respect to $X_k$, then $X_k$ is almost surely not effective.
 \end{prop}

\begin{proof}
Given $x,x'\in \mathbb{R}$, CATEs are given by:

\begin{align}
CATE(X_i=x)=&\sum_{i=1}^q y_i [p_i(x;z)-p_i(x;z')]=\sum_{i=1}^q y_i p_i(x;z,z')\\ 
 CATE(X_i=x')=&\sum_{i=1}^q y_i [p_i(x';z)-p_i(x';z')]=\sum_{i=1}^q y_i p_i(x';z,z').
\end{align}

We now will prove the proposition by contrapositive. Suppose that $X_k$ is effective. If so, then there exists an index set, $D$, with at least two elements such that $CATE(x)-CATE(x')=\sum_{i \in D} y_j [p_i(x;z,z')-p_i(x';z,z')]$ and any $p_i(x;z,z')=0$ for all $i \notin D$. Because $y_i$ is arbitrarily set and is independent of $p_i$, and $Y$ has absolutely continuous distribution, the probability that $\sum_{j \in D} y_j [p_j(x;z,z')-p_j(x';z,z')]=0$ is zero.

\end{proof}

We use the following example to illustrate the above proposition.

\begin{example}
Suppose $h(Y)$ has the following form:

$$
h(Y) = 
\begin{cases}
	y_1 & Y \in (-\infty,c_1]\\
	y_2 & Y \in (c_1,\infty)\\
\end{cases}
$$ where $Y=h(X_1,X_2,Z)$.\\

Then, let us calculate the CATE $X_2$ , given $z,z' \in Z$:

$$
\begin{aligned}
	CATE(X_2=x) &= y_1 [p_1(x;z)-p_1(x;z')] + y_2[p_2(x;z)-p_2(x;z')]\\
	&= y_1 p_1(x;z,z') + y_2 p_2(x;z,z')
\end{aligned}
$$

and 

$$
\begin{aligned}
	CATE(X_2=x') &= y_1 [p_1(x';z)-p_1(x';z')] + y_2[p_2(x';z)-p_2(x';z')]\\
	&= y_1 p_1(x';z,z') + y_2 p_2(x';z,z')
\end{aligned}
$$

If $X_k$ is not effective, then $CATE(X_2=x)=CATE(X_2=x')$, therefore there exist no HTE.
If  $X_k$ is effective, then non-existence of HTEs requires that 

$$
y_1 p_1(x;t,t') + y_2 p_2(x;z,z') = y_1 p_1(x';z,z') + y_2 p_2(x';z,z')
$$

\begin{equation}\label{equ:exp3}
	\frac{y_1}{y_2}=\frac{p_2(x';z,z')-p_2(x;z,z')}{p_1(x;z,z')-p_1(x';z,z')}
\end{equation}

For arbitrarily chosen $y_1\in \mathbb R$ and $y_2 \in \mathbb R$, the above equality holds with probability zero if $p_1$ or $p_2$ can take value in a set with Lebesgue measure larger than 0.

\end{example}

\section[ Examples from the Literature]{Examples from the literature}\label{sec:examples}

\subsection{Description of study context, design, measures} \label{ss:context}
Table \ref{tab:studies} provides a summary of the context, design, and measures used in the empirical studies detailed in Section \ref{sec:guidance}. Note that the \citet{ariasetal2019} does not focus on heterogeneous treatment effects by partisan affiliation, though partisan affiliation (bias) is a parameter in their model. 

\begin{table}
\resizebox{\textwidth}{!}{
	\begin{tabular}{l|p{4cm}p{4cm}p{4cm}p{4cm}} \hline
	& \citet{eggers2014} & \citet{ariasetal2019} & \citet{anduizaetal2013} & \citet{defigueredoetal2023} \\ \hline \hline 
Context & UK parliamentary elections in 2010 & Mexican municpal elections in 2015 & Spain 2010 (survey) & Sao Paulo, Brazil mayoral runoff in 2008 \\ \\
Design & Observational study & Field experiment & Survey experiment & Field, survey experiments \\ \\
Treatment &  Binary indicator of sufficient implication in 2009 expenses scandal coded from media reports& Dissemination of information on municipal use of intergovernmental infrastructure transfer in municipalities with high or low malfeasant spending & Newspaper article-style vignette about corrupt mayor that randomizes partisan affiliation. & Flyer distribution about one of two candidates in runoff having convictions involving impropriety in office (field experiment). Same information distributed to untreated voters after election (survey experiment). \\

Outcome(s) & Incumbent vote share (constituency) and British Election Survey self-reported vote for incumbent & Precinct-level incumbent party vote share. & Perceived seriousness of  corruption allegation on 1-10 scale. & Precinct-level voting outcomes for each candidate (field experiment). Rating of support for each candidate on 0-10 scale (survey experiment).\\

Partisanship & Indicator for Labour-Conservative constituency.  & (Not measured directly) & Respondent closeness/affiliation to party in vignette.  & Pre-treatment support for runoff candidate \\ \hline

	\end{tabular}

}
\caption{Design of the four studies discussed in Section \ref{sec:guidance}. Note that some studies include additional outcome variables or different moderators (instead of or in addition to partisanship).}\label{tab:studies}
\end{table}
\subsection{Justification of theoretical representations} \label{ss:theory}
The four examples that we use in Section \ref{sec:guidance} develop and communicate their theories in different ways. We attempt to remain as faithful as possible to the underlying theories when translating these theoretical assumptions into the models we present in Figure \ref{fig:dags}. We describe these ``translations'' below (with documentation from the original studies):
\begin{enumerate}
\item \citet{eggers2014}: Presents a formal framework to illustrate the argument. In the framework, a voter ($i$) will vote for the incumbent from party $p_I$ over challenger from party $p_C$ if:
\begin{align*}
u^i(p_I)-u^i(p_C) \geq c_I^{(i)}, &&\text{Equation (1), p. 445}
\end{align*}
where $u^i(\cdot)$ is the utility of having a politician of a given party in power and $c_I$ is the incumbent ``candidate's perceived corruption level'' (p. 445). The individual-level superscript on $c_I^{(i)}$ is our addition.\\

Translating to our notation, $a_i \equiv u^i(p_I)-u^i(p_C)$ measures alignment with the incumbent party. Suppose further that $c_I^i$ may at the individual level as a function of individual exposure to information and an aversion to corruption. While we use the parameter $c_i$ slightly differently, in this case: 
$$\underbrace{c_I^i}_{\text{\citet{eggers2014}}} \equiv \underbrace{c_i \lambda_i}_{\text{Our notation}}$$
This framework most closely aligns with the model we present as a running example, though it omits a valence shock.  
\item \citet{ariasetal2019} present a formal model in which the voter's utility is given by their equations (1) and (3):

\begin{align*}
u_i(I \mid s) &= \underbrace{-(\pi^1(s) \theta^H + [1-\pi^1(s)]\theta^L)}_{\text{Expressive disutility from party }p} + \underbrace{\delta_i}_{\text{Partisan bias}} + \underbrace{I[s \in \{l, h\}]I[\Delta e(s)\leq 0]\sum_{j \in N_i}I[\Delta u_j(I\mid s)\geq 0]}_{\text{Coordination motive}}
\end{align*}
Here, the signal, $s$, is corruption information ($c_i$ in our notation). Given the assumption of common prior that the incumbent is a high type ($\theta^H)$, the expressive disutility on the basis of the posterior $\pi^1(s)$ is analogous to our distaste for incumbent corruption $\lambda_i c_i$. They include a second mechanism, the voter's coordination motives, which depends on the signal ($s$). This is absent from our motivating model. Partisan bias, $\delta_i$, is assumed to be a uniformly-distributed random variable. In our model, this is analogous to $a_i$. Importantly, partisan bias is not assumed to moderate either mechanism in our model or in theirs. Further, the coordination and distaste mechanisms are assumed to be additively separable with respect to utility.

\item \citet{anduizaetal2013} present a verbal argument about the possible channels through which partisanship might influence assessments of a candidate's corruption. They argue that:
\begin{quote}
``The \emph{limited information hypothesis} (or the ignorant voter in its earliest version) argues that citizens vote for corrupt politicians because they are not sufficiently informed about their misbehaviors. If people knew with certitude that an official was corrupt, they would evaluate that person negatively and vote against him or her.'' (p. 1677).
\end{quote}
This suggests that provision of information about corruption ($c_i$) should link to a citizens' assessment of incumbent corruption. But, they note that:
\begin{quote}
``The literature suggests two alternative explanations: citizens do not give \emph{credibility} to this information or they do not consider it \emph{important} enough compared with other elements when they think about politics.'' (p. 1677)
\end{quote}
The importance qualification for the effect of information provision connects to corruption aversion (i.e. $\lambda_i$)---how much does corruption matter to the voter? The credibility qualification is subsequently associated with motivated reasoning:
\begin{quote}
``Downplaying the importance of corruption when it affects the own party is to be expected as it reduces cognitive dissonance. It is well known that people tend to see the political world in a way that is consistent with their political predispositions$\ldots$This argument is consistent with the theory of motivated information processing'' (p. 1688).
\end{quote}
Motivated information processing/reasoning is subsequently connected to partisan affiliation, $a_i$ as a moderator of the effect of information on assessments of the incumbent's corruption.
\item \citet{defigueredoetal2023} present a verbal argument that first emphasizes the effect of information on voting behavior:
\begin{quote}
``An important precondition for electoral accountability is whether or not voters have access to information about the corrupt behavior of public officials, which may prompt them to vote
against such candidates on Election Day.'' (p. 729)
\end{quote}
This justifies the link between information disclosure ($c_i$) and assessments of incumbent corruption. They subsequently argue that partisan bias can condition the effects of information on voting behavior:
\begin{quote}
``Partisan bias may function as a buffer against corruption scandals and hinder electoral accountability for voters in democracies with established parties.'' (p. 729)
\end{quote}
They argue that partisan bias shapes corruption aversion, which in turn leads to variation in response to corruption information and subsequent voter behavior.
\begin{quote}
``Relatively young parties that built their reputation among voters on an anti-corruption platform may thus suffer harsher punishment from sympathizers when compared to other traditional parties.'' (p. 730)
\end{quote}

\end{enumerate}

\section[  Strengthening Assumptions]{Strengthening Assumptions}\label{app:sol}

In the main text, Propositions \ref{prop3} and \ref{prop4} indicate that the existence or non-existence of HTEs are not generally informative about mechanism activation. In this section, we explore the conditions under which invoking stronger assumptions can provide more information. \\

Recall the basic problem with transformed outcomes. In Figure \ref{fig:inddgp}, $h(Y)$ is the non-linear transformation of $Y$. Our main result shows that $X_{k+1}$ can also induce HTEs even though $X_{k+1}$ does not detect the mechanism. \\

\begin{figure}[!h]
	\begin{center}
		\begin{tikzpicture}
			\node at (0,0) (z) {$Z$};
			\node at (2,0) (m) {$M$};
			\node at (3.5, 0) (u) {$Y$};
			\node at (5, 0) (y) {$h(Y)$};
			\node at (1,1) (x1) {$X_{k}$};
			\node at (2,-1) (x2) {$X_{k+1}$};
			\draw[->] (z)--(m);
			\draw[->] (x1)--(1, 0.05);
			\draw[->] (m)--(u);
			\draw[->] (u)--(y);
			\draw[->] (x2)--(m);
		\end{tikzpicture}
	\end{center}
	\caption{DGP for a transformed outcome, $h(Y)$.} \label{fig:inddgp}
\end{figure}

\subsection[Transformed Outcomes under Monotonicity Assumptions]{Transformed Outcomes under Monotonicity Assumptions}\label{app:monotonicty}

In practice, people frequently and implicitly assume monotonicity of treatment effects. We ask whether this assumption permits inference about mechanism activation for transformed outcomes of interest. To be specific, consider the two DGPs in figure \ref{fig:monotone}. We will index these DGPs by $s \in \{1, 2\}$ where $s = 1$ corresponds to the left DAG and $s = 2$ corresponds to the right DAG.

\begin{figure}[!h]
	\begin{center}
		\begin{tikzpicture}
			\node at (0,0) (z) {$Z$};
			\node at (2,0) (m) {$M$};
			\node at (3.5, 0) (u) {$Y_1$};
			\node at (5, 0) (y) {$h(Y_1)$};
			\node at (1,1) (x1) {$X$};
			\draw[->] (z)--(m);
			\draw[->] (x1)--(1, 0.05);
			\draw[->] (m)--(u);
			\draw[->] (u)--(y);
			
			\node at (7,0) (z1) {$Z$};
			\node at (9,0) (m1) {$M$};
			\node at (10.5, 0) (u1) {$Y_2$};
			\node at (12, 0) (y1) {$h(Y_2)$};
			\node at (9,1) (x2) {$X$};
			\draw[->] (z1)--(m1);
			\draw[->] (x2)--(9, 0.2);
			\draw[->] (m1)--(u1);
			\draw[->] (u1)--(y1);
		\end{tikzpicture}
	\end{center}
	\caption{Two different DGPs. On the left, in DGP 1, $X$ is an MDV. On the right, in DGP 2, $X$ is not an MDV.} \label{fig:monotone}
\end{figure}

The left panel assumes $X_k$ is an MDV. $X_{k+1}$ is not an MDV in the right panel. In the figure, there are no other mediators. Therefore, both graphs satisfy exclusion Assumptions \ref{cond:irre1} and \ref{cond:irre2} by construction. \\

We will assume that $Y_s$ is a latent directly-affected outcome and $h(Y_s)$  is a binary variable given by:

\begin{align}\label{equ:dgpzo1}
	h(Y_s) &= \begin{cases}
		0 & Y_s \in (-\infty, c]\\
		1 & Y_s \in (c, \infty],
	\end{cases}
\end{align}

for some $c \in (-\infty, \infty)$. Propositions \ref{prop3} and \ref{prop4} show that we cannot differentiate between the DGPs 1 and 2 in Figure \ref{fig:monotone} on the basis of the existence or non-existence of HTEs for outcome $h(Y_s)$.\\

We will consider what can be gained by imposing a monotonicity assumption of the form: $\frac{\partial^2 Y}{\partial Z \partial X} > (<) 0$ (note that the inequalities are strict).\footnote{Writing the monotonicity assumption in this way assumes that this derivative exists.} Clearly montonicity can hold in the left panel (where $X$ is an MDV) but $\frac{\partial^2 Y}{\partial Z \partial X} = 0$ in the right panel in which $X$ is not an MDV. To explore the implications of monotonicity, we consider following DGPs for $s=\{1,2\}$:

\begin{align}\label{equ:dgpzo2}
Y_i= g_i(Z,X)+e_i
\end{align}

For $X \in X^{MDV}$, e.g., the DGP 1 in the left panel of Figure \ref{fig:monotone}, suppose that montonicity holds such that  $\frac{\partial^2 Y_1}{\partial X \partial Z}:=\beta(x,z)$, where $\beta(x,z)$ is either strictly positive or negative. For $X \notin X^{MDV}$, e.g., the right DGP in Figure \ref{fig:monotone}, by definition, we have $\frac{\partial^2 Y_2}{\partial x \partial z}=0$. \\

We ask whether researchers can differentiate these two cases when motonicity holds for the first DGP (e.g., an assumption of monotonicity).  First, given \eqref{equ:dgpzo1} and \eqref{equ:dgpzo2}, note that:

\begin{align}
E[h(Y_s)|Z, X] = \Pr(e_i \ge c-g_s(Z,X))
\end{align}

Let $f_e$ be the density of $e$ denote its derivative as $f'_e$.  We can then express HTEs for $h(Y_1)$ as: 

\begin{align}\label{equ:dgp1}
- f'_e(c-g_1)\frac{\partial g_1}{\partial X} \frac{\partial g_1}{\partial Z}+f_e(c-g_1)\beta(x,z)
\end{align}

and the HTEs for $h(Y_2)$ is 

\begin{equation}\label{equ:dgp2}
- f'_e(c-g_2)\frac{\partial g_2}{\partial X} \frac{\partial g_2}{\partial Z}
\end{equation}

The additional term $-f_e(c-g_1)\beta(x,z)$ in equation \eqref{equ:dgp1} may help us to differentiate two DGPs by generating a differently-signed HTE. 

~\\

\textbf{Sign Differences}

If, under certain $x$ and $z$, \eqref{equ:dgp2} and the first term of \eqref{equ:dgp1} have the same sign and the second term of \eqref{equ:dgp1} has the opposite sign and is sufficiently large, \eqref{equ:dgp2} and \eqref{equ:dgp1} will have different signs.  Moreover, If $e$ is uniformly distributed, then $f'_e=0$ and thus equation \eqref{equ:dgp2} is equal to 0 while equation \eqref{equ:dgp1} is non-zero.

We summarize the discussion in the following proposition.

%We will see that whether this strategy works highly depend on the behavior of $f'_e(\cdot)$. Therefore, to simplify the analysis, we assume the distribution of $e$ is unimodal, which means that $f'_e(u)$ is  positve with small $u$ and negative with large $u$. 

\begin{prop}\label{prop:monotone}
	Consider the transformed outcome $h(Y_s)$ satisfying equation \eqref{equ:dgpzo1} in which moderation effect $\beta(x, z)$ is monotonic.  \\
	
	(1) Suppose that $e$ is uniformly distributed, then HTEs for $h(Y_2)$ is 0.\\
	
	(2)  Suppose $e$ is not uniformly distributed, then HTEs for $h(Y_1)$ and $h(Y_2)$ have different signs under two cases:\\
	\begin{itemize}
	\item[](2a)  $f'_e(c-g_2) \frac{\partial g}{\partial X}\frac{\partial g}{\partial Z}<0$ and 
$\beta(x,z) < \frac{f'_e(c-g_1)\frac{\partial g_1}{\partial X} \frac{\partial  g_1}{\partial Z}}{f_e(c-g_1)}$; or

	\item[](2b) $f'_e(c-g_2) \frac{\partial g}{\partial X}\frac{\partial g}{\partial Z}>0$ and 
	$\beta(x,z) >\frac{f'_e(c-g_1)\frac{\partial g_1}{\partial X} \frac{\partial  g_1}{\partial Z}}{f_e(c-g_1)}$
\end{itemize}
\end{prop}

%\begin{proof}
%	(1)Note that if $\epsilon$ is uniformly distributed, then $f'_e=0$.
	
  % (2) It directly follows equation \ref{equ:dgp1} and \ref{equ:dgp2}.
%\end{proof}

In practice, however, it is difficult to verify conditions in (2). Corollary \ref{ap:corr} provides additional assumptions on $g(\cdot)$ and/or the tail behavior of $e$ that are sufficient to satisfy these conditions.

\begin{corollary}\label{ap:corr}
 Suppose conditions in proposition \ref{prop:monotone} holds. Assume that:  \\
 
 (a) $g$ is increasing in $X$ and $Z$,\\
 
 (b) the distribution of $e$ is unimodal,\\
 
 (c) $\beta$ is increasing in $X$ and $Z$,\\
 
 then \\
 
 (1)  small values of $X$ and $Z$ satisfy condition (2a), if any such $x,z$ exists;\\
 
 (2) large values of $X$ and $Z$ satisfy condition (2b), if any such $x,z$ exists.

\end{corollary}
\begin{proof}
It is straightforward to prove the corollary. If we pick small values of $x$ and $z$ in the data, then by condition (1) $g$ is small and $\frac{\partial g}{\partial X}\frac{\partial g}{\partial Z}>0$, and by (2) $f'_e(c-g)<0$, by (3) $\beta$ is small enough as well. These together imply 2(a) in proposition \ref{prop:monotone} is satisfied. The same logic holds for (2).
\end{proof}

Proposition \ref{prop:monotone} reveals and Corollary \ref{ap:corr} illustrates that Assumptions \ref{cond:irre1}-\ref{cond:irre2} and monotonicity are not, in general, jointly sufficient to use HTEs to assess mechanism activation on tansformed outcomes.

%\subsubsection{Approach 2: Effect Size}

%From \eqref{equ:dgp1} and \eqref{equ:dgp2}, we see the additional term $f_e(c-g_1)\beta(x,z)$ can result in different HTE effect sizes of two DAGs. In particular, we may expect $X^{MDV}$ to induce a larger effect than $X^{non-MDV}$. For example, in practice, given a set of candidate covariates $\{X_i, X_2,..., X_k\}$, we find one covariate generates quite larger HTE than others or (thus) is the only statistically significant covariate, we can conclude it is $MDV$.

\subsection[Assumptions on the Latent Utility Distribution]{ Assumptions on the Latent Utility Distributions}\label{app:transform}

In practice, one of the most common cases is that $Y$ is the latent utility, and $h(Y)$ is the observed action or discrete choice. If we can re-construct the utility from the observed data, then we can use it as the new outcome and use HTEs to assess mechanism activation under Assumptions 1-2. From Propositions \ref{prop1} and \ref{prop2}, we know more information on mechanism activation can be ascertained from HTEs on the latent outcome.\\

To re-construct the utility, one popular solution is to apply random utility models (RUM). In such a model, a decision maker, $i$, faces a choice among $M$ alternatives. The utility that decision maker $i$ obtains from alternative $m$ is $U^i_{m}$. The decision maker $i$ is assumed to choose alternative $m$ if and only if $U^i_{m} > U^i_{m'} \forall m' \neq m$. The researcher does not observe the decision maker's utility. We instead observe only attributes of the alternatives and decision makers. A function $V$ can be specified with those observed attributes to relate to the decision maker's utility. Therefore, the utility is decomposed as $U^i_m = V^i_m + \varepsilon^i_m$, where $\varepsilon^i_m$ captures the unobserved factors that affect utility.\\

The most widely used assumption in RUMs is that $\varepsilon$ is independently, identically distributed according to an extreme value density function. The density function is $f(\varepsilon^i_m) = e^{-\varepsilon^i_m}e^{-e^{-\varepsilon^i_m}}$ and the CDF is $F(\varepsilon^i_m)=e^{-e^{\varepsilon^i_m}}$. The difference between two extreme value variables is distributed logistic: let $\varepsilon^i = \varepsilon^i_m-\varepsilon^i_{m'}$, then

$$
F(\varepsilon^i)=\frac{e^{\varepsilon^i}}{1+e^{\varepsilon^i}}
$$

The extreme value distribution (and thus logistic distribution) is similar to normal but has fatter tails. Accordingly, we get the familiar logit choice probability that the decision maker chooses alternative $m$:

$$
\frac{e^{V^i_m}}{\sum_{m}e^{V^i_m}}
$$

In practice, utility is usually specified as a linear function $V^i_m=X \beta$, where $X$ is the vector of observed variables and $\beta$, are parameters. We can estimate $\beta$ from the data. Then, we could treat $V_m^i$ as the new outcome to explore the mechanism using HTE. It should be clear that Propositions \ref{prop1} and \ref{prop2} hold for $V_i^m$ even if the utility is $U$ by linearity.\\

Now, we use our motivating example to illustrate how to apply the above RUM and distribution assumptions in practice. In the main text, we provide a simple model of the effect of corruption information on pro-incumbent voting. The model specifies the systematic component of utility, $V_m^i$. A RUM requires us to specify a random component of utility, $\varepsilon$ which also affects individual behavior and is additively separable from $V_m^i$ (i.e. not relevant to the theory of interest). Therefore,  given the observed voting data we have, if we assume the distribution of $\varepsilon$ is type-I extreme value, then the probability individual votes for the incumbent is 

$$
\mathbb{P}(Y_{i2}=1) = \frac{e^{V_m^i}}{1+e^{V_m^i}}
$$

In practice, researchers tend to specify a linear model to approximate $V_m^i$, though other functional forms are also possible. If $V_m^i$ is correctly specified, we then treat $V_m^i$ as an outcome (utility) and use this measure when estimating HTE. %(How to do inference? need further research.)

\subsection{Assumptions about the Magnitude of HTEs}\label{sec:magnitude}

In practice, the exclusion assumptions may be too restrictive (strong). Consider the scenario depicted in the figure below, where our focus is on the activation of mechanism $M_1$. However, the covariate $X_k$ may influence mechanism $M_2$ and/or the direct effect. Our results indicate that we cannot learn about the activation of $M_1$ from the existence of HTE in $X_k$ since the exclusion assumption is violated. Here, we consider what might be learned from observing a large (in magnitude) HTE. This intuition can be formalized using Bayesian inference.\\

\begin{figure}[!h]
	\begin{center}
		\begin{tikzpicture}
			\node at (0,0) (z) {$Z$};
			\node at (3,1.1) (m) {$M_1$};
			\node at (3,0) (m2) {$M_2$};
			\node at (6, 0) (y) {$Y_1$};
			\node at (8, 0) (y2) {$Y_2$};
			\node at (1,2.5) (x1) {$X_k$};
			\draw[->] 
			(z) edge [bend left = 22.5] (m)
			(m) edge [bend left = 22.5] (y)
			(z) edge [bend right = 35] (y);
			\draw[->, dashed, color = blue]
			(x1) edge [bend right = 20] (1, -.5);
			\draw[->, dash dot, color = red]
			(x1) edge [bend left = 20] (2, 0.1);
			\draw[->, dash dot, color = red]
			(x1) edge [bend left = 20] (4, 0.1);
			
			\draw[->] 
			(x1) edge [bend left = 30] (y);
			\draw[->] 
			(x1) edge [bend left = 10] (m2);
			\draw[->, color= orange, thick] (x1)--(1.25, .95);
			\draw[->, color = orange, thick] 
			(x1) edge [bend left = 20] (4, 1.15);
			\draw[->] (z)--(m2);
			\draw[->] (m2)--(y);
			\draw[->] (y)--(y2);
		\end{tikzpicture}
	\end{center}
	\caption{The blue and yellow dashed arrows represent potential exclusion assumption violations.}
\end{figure}

Suppose that we have a prior $p$ over the activation of mechanism $M_1$. We characterize the violation of the exclusion assumptions---e.g., HTEs (the magnitude of the difference of CATEs) in $X_k$ that are generated by other mechanisms---with a normal distribution, $\mathcal{N}(\mu_2, \sigma_2^2)$. The magnitude  $|\mu_2|$ captures the severity of the violation. We assume that the covariate $X_k$ generates greater heterogeneity via the mechanism of interest, $M_1$, which we capture with the normal distribution $\mathcal{N}(\mu_1, \sigma_1^2)$, where $|\mu_1|>|\mu_2|$. Then, the posterior after we observe such a large HTE is given by 
\begin{align*}
	\mathbb{P}[M_1 \text{ is active}\mid\text{as extreme as  HTE}] & = \frac{p\mathbb{P}[\text{HTE}|\text{active}]}{p\mathbb{P}[\text{HTE}|\text{active}] + (1-p)\mathbb{P}[\text{HTE}|\text{not active}]}\\
	&=\frac{p\Phi( -2|\frac{HTE-\mu_1-\mu_2}{\sigma_1+\sigma_2}| )}{p \Phi( -2|\frac{HTE-\mu_1-\mu_2}{\sigma_1+\sigma_2}| ) + (1-p)\Phi( -2|\frac{HTE-\mu_2}{\sigma_2}| )}
\end{align*} where $\Phi(\cdot)$ is the CDF of standard normal distribution. \\

A Bayesian analyst might specify a prior distribution for the activation of a mechanism, denoted as $p \in [0,1]$, using a Beta distribution, $Beta(a,b)$. Subsequently, based on observed data, the posterior density of $p$ is \begin{align*}
	\mathbb{P}[p|\text{HTE}]  = \frac{dBeta(a,b) [\theta \Phi( -2|\frac{HTE-\mu_1-\mu_2}{\sigma_1+\sigma_2}| ) + (1-\theta) \Phi( -2|\frac{HTE-\mu_2}{\sigma_2}| ) ]}{\int_0^1dBeta(a,b) [\theta \Phi( -2|\frac{HTE-\mu_1-\mu_2}{\sigma_1+\sigma_2}| ) + (1-\theta) \Phi( -2|\frac{HTE-\mu_2}{\sigma_2}| ) ]dp.}
\end{align*}  

For readers well-versed in Bayesian statistics, the posterior distribution of $p$ can be characterized using a variety of sampling techniques. For less familiar readers, we have provided a straightforward \texttt{R} script, which requires only the \texttt{ggplot2} package and is designed to be user-friendly. It automatically generates figures like Figure \ref{fig:bayes} to illustrate both the prior and posterior distributions, where we set $HTE=5$, $\mu_1=1$, $\mu_2=0.5$, $Beta(2,3)$, and $\sigma^2_1=\sigma^2_2=1$ below.\\

\begin{figure}[!h]
	\centering
	\includegraphics[width=0.5\linewidth]{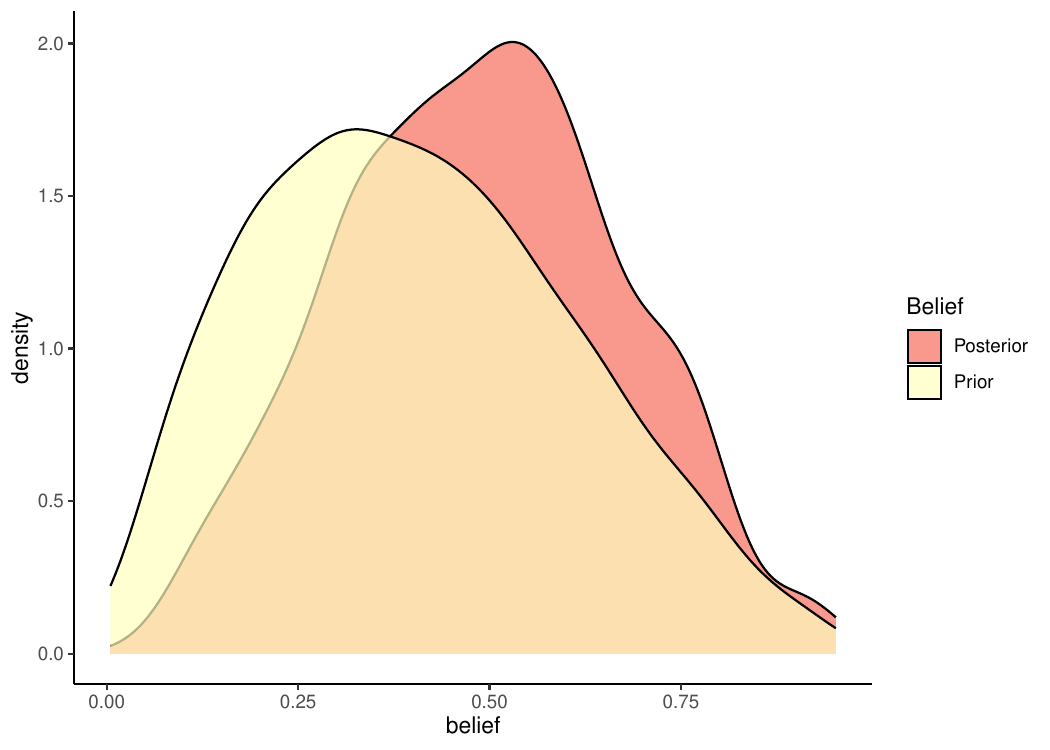}
	\caption{Illustration of prior and posterior beliefs under a model in which $HTE=5$,$\mu_1=1$, $\mu_2=0.5$, $Beta(2,3)$, and $\sigma^2_1=\sigma^2_2=1$. The densities reflect the prior and posterior distributions of $p$, the probability that mechanism $1$ is active. The shift in the distribution to the right from prior to posterior suggests that the large HTE strengthens our belief that mechanism $1$ is active.}
	\label{fig:bayes}
\end{figure}

The strategy described above can also be applied to conduct sensitivity analysis. For instance, starting with a non-informative prior, one can determine how large $\mu_2$ needs to be---measured as a proportion of the total Average Treatment Effect (ATE)---to result in a posterior probability of the mechanism being active that is less than $50\%$.

\clearpage
\singlespacing
\bibliographystyle{apsr2}
\bibliography{literature}

\begin{thebibliography}{xx}

\harvarditem[Anduiza, Gallego, \harvardand\ Mu{\~{n}}oz]{Anduiza, Gallego,
  \harvardand\ Mu{\~{n}}oz}{2013}{anduizaetal2013}
Anduiza, Eva, Aina Gallego, \harvardand\ Jordi Mu{\~{n}}oz. 2013.
\newblock ``Turning a Blind Eye: Experimental Evidence of Partisan Bias in
  Attitudes Toward Corruption.'' {\em Comparative Political Studies} 46 (12):
  1664--1692.

\harvarditem[Arias et~al.]{Arias, Bal\'{a}n, Larreguy, Marshall, \harvardand\
  Querub\'{i}n}{2019}{ariasetal2019}
Arias, Eric, Pablo Bal\'{a}n, Horacio Larreguy, John Marshall, \harvardand\
  Pablo Querub\'{i}n. 2019.
\newblock ``Information Provision, Voter Coordination, and Electoral
  Accountability: Evidence from Mexican Social Networks.'' {\em {American
  Political Science Review}} 113 (2): 475--498.

\harvarditem[{de Figueiredo}, Hidalgo, \harvardand\ Kasahara]{{de Figueiredo},
  Hidalgo, \harvardand\ Kasahara}{2023}{defigueredoetal2023}
{de Figueiredo}, Manuel~F.P., {F. Daniel} Hidalgo, \harvardand\ Yuri Kasahara.
  2023.
\newblock ``When Do Voters Punish Corrupt Politicians? Experimental Evidence
  from a Field and Survey Experiment.'' {\em British Journal of Political
  Science} 53: 728--739.

\harvarditem{Eggers}{2014}{eggers2014}
Eggers, Andrew~C. 2014.
\newblock ``Partisanship and Electoral Accountability: Evidence from the UK
  Expenses Scandal.'' {\em Quarterly Journal of Political Science} 9: 441--472.

\harvarditem[Haim, Ravanilla, \harvardand\ Sexton]{Haim, Ravanilla,
  \harvardand\ Sexton}{2021}{haimetal2021}
Haim, Dotan, Nico Ravanilla, \harvardand\ Renard Sexton. 2021.
\newblock ``Sustained Government Engagement Improves Subsequent Pandeic Risk
  Reporting in Conflict Zones.'' {\em {American Political Science Review}} 115
  (2): 717--724.

\harvarditem[Imai, Keele, \harvardand\ Tingley]{Imai, Keele, \harvardand\
  Tingley}{2010}{imai2010general}
Imai, Kosuke, Luke Keele, \harvardand\ Dustin Tingley. 2010.
\newblock ``A general approach to causal mediation analysis.'' {\em
  Psychological methods} 15 (4): 309.

\harvarditem{Imai \harvardand\ Yamamoto}{2013}{imai2013identification}
Imai, Kosuke, \harvardand\ Teppei Yamamoto. 2013.
\newblock ``Identification and sensitivity analysis for multiple causal
  mechanisms: Revisiting evidence from framing experiments.'' {\em Political
  Analysis} 21 (2): 141--171.

\harvarditem{Moscowitz}{2021}{moscowitz2021}
Moscowitz, Daniel. 2021.
\newblock ``Local News, Information, and the Nationalization of U.S.
  Elections.'' {\em {American Political Science Review}} 115 (1): 114--129.

\end{thebibliography}


\begin{thebibliography}{}

\bibitem[Abramson et~al., 2022]{abramson2022we}
Abramson, S.~F., Ko{\c{c}}ak, K., \& Magazinnik, A. (2022).
\newblock What do we learn about voter preferences from conjoint experiments?
\newblock {\em American Journal of Political Science}, 66(4), 1008--1020.

\bibitem[Acharya et~al., 2016]{acharya2016explaining}
Acharya, A., Blackwell, M., \& Sen, M. (2016).
\newblock Explaining causal findings without bias: Detecting and assessing
  direct effects.
\newblock {\em American Political Science Review}, 110(3), 512--529.

\bibitem[Anduiza et~al., 2013]{anduizaetal2013}
Anduiza, E., Gallego, A., \& Mu{\~{n}}oz, J. (2013).
\newblock Turning a blind eye: Experimental evidence of partisan bias in
  attitudes toward corruption.
\newblock {\em Comparative Political Studies}, 46(12), 1664--1692.

\bibitem[Arias et~al., 2019]{ariasetal2019}
Arias, E., Bal\'{a}n, P., Larreguy, H., Marshall, J., \& Querub\'{i}n, P.
  (2019).
\newblock Information provision, voter coordination, and electoral
  accountability: Evidence from mexican social networks.
\newblock {\em {American Political Science Review}}, 113(2), 475--498.

\bibitem[Ashworth et~al., 2021]{ashworth2021theory}
Ashworth, S., Berry, C.~R., \& Bueno~de Mesquita, E. (2021).
\newblock {\em Theory and Credibility: Integrating Theoretical and Empirical
  Social Science}.
\newblock Princeton University Press.

\bibitem[Ashworth et~al., 2023]{ashworthetal2023}
Ashworth, S., Berry, C.~R., \& de~Mesquita, E.~B. (2023).
\newblock Modeling theories of women's underrepresentation in elections.
\newblock {\em American Journal of Political Science}, Early View.

\bibitem[Athey et~al., 2019]{atheyetal2019}
Athey, S., Tibshirani, J., \& Wager, S. (2019).
\newblock Generalized random forests.
\newblock {\em Annals of Statistics}, 47(2), 1148--1178.

\bibitem[Athey \& Wager, 2021]{atheywager2021}
Athey, S. \& Wager, S. (2021).
\newblock Policy learning with observational data.
\newblock {\em Econometrica}, 89(1), 133--161.

\bibitem[Berry et~al., 2009]{berryetal2009}
Berry, W.~D., DeMeritt, J. H.~R., \& Esarey, J. (2009).
\newblock Testing for interaction in binary logit and probit models: Is a
  product term essential?
\newblock {\em American Journal of Political Science}, 54(1), 248--266.

\bibitem[Blackwell et~al., 2024]{blackwelletal2024}
Blackwell, M., Ma, R., \& Opacic, A. (2024).
\newblock Assumption smuggling in intermediate outcome tests of causal
  mechanisms assumption smuggling in intermediate outcome tests of causal
  mechanisms.
\newblock Working paper available at \url{arXiv:2407.07072v2}.

\bibitem[Brambor et~al., 2006]{bramboretal2017}
Brambor, T., Clark, W.~R., \& Golder, M. (2006).
\newblock Understanding interaction models: Improving empirical analyses.
\newblock {\em Political Analysis}, 14(1), 63--82.

\bibitem[Bueno~de Mesquita \& Tyson, 2020]{bdmtyson2020}
Bueno~de Mesquita, E. \& Tyson, S.~A. (2020).
\newblock The commensurability problem: Conceptual difficulties in estimating
  the effect of behavior on behavior.
\newblock {\em {American Political Science Review}}, 114(2), 375--391.

\bibitem[Bullock \& Green, 2021]{bullockgreen2021}
Bullock, J.~G. \& Green, D.~P. (2021).
\newblock The failings of conventional mediation analysis and a design-based
  alternative.
\newblock {\em Advances in Methods and Practices in Psychological Science},
  4(4), 1--18.

\bibitem[{de Figueiredo} et~al., 2023]{defigueredoetal2023}
{de Figueiredo}, M.~F., Hidalgo, F., \& Kasahara, Y. (2023).
\newblock When do voters punish corrupt politicians? experimental evidence from
  a field and survey experiment.
\newblock {\em British Journal of Political Science}, 53, 728--739.

\bibitem[Devaux \& Egami, 2022]{devauxegami2022}
Devaux, M. \& Egami, N. (2022).
\newblock Quantifying robustness to external validity bias.
\newblock Working paper available at
  \url{https://naokiegami.com/paper/external_robust.pdf}.

\bibitem[Dunning et~al., 2019]{dunningetal2019}
Dunning, T., Grossman, G., Humphreys, M., Hyde, S.~D., McIntosh, C., \& Nellis,
  G., Eds. (2019).
\newblock {\em Information, Accountability, and Cumulative Learning: Lessons
  from Metaketa I}.
\newblock New York: {Cambridge University Press}.

\bibitem[Egami \& Hartman, 2022]{egamihartman2020}
Egami, N. \& Hartman, E. (2022).
\newblock Elements of external validity: Framework, design, and analysis.
\newblock {\em {American Political Science Review}}, Forthcoming.

\bibitem[Eggers, 2014]{eggers2014}
Eggers, A.~C. (2014).
\newblock Partisanship and electoral accountability: Evidence from the uk
  expenses scandal.
\newblock {\em Quarterly Journal of Political Science}, 9, 441--472.

\bibitem[Ferraz \& Finan, 2008]{ferrazfinan2008}
Ferraz, C. \& Finan, F. (2008).
\newblock Exposing corrupt politicians: The effects of brazil's publicly
  released audits on electoral outcomes.
\newblock {\em Quarterly Journal of Economics}, 123(2), 703--745.

\bibitem[Fink et~al., 2014]{finketal2014}
Fink, G., McConnell, M., \& Vollmer, S. (2014).
\newblock Testing for heterogeneous treatment effects in experimental data:
  falsediscovery risks and correction procedures.
\newblock {\em Journal of Development Effectiveness}, 6(1), 44--57.

\bibitem[Fu, 2024]{fu2024extracting}
Fu, J. (2024).
\newblock Extracting mechanisms from heterogeneous effects: An identification
  strategy for mediation analysis.
\newblock {\em arXiv preprint arXiv:2403.04131}.

\bibitem[Gerber \& Green, 2012]{gerbergreen2012}
Gerber, A.~S. \& Green, D.~P. (2012).
\newblock {\em Field Experiments: Design, Analysis, and Interpretation}.
\newblock New York: W.W. Norton.

\bibitem[Grimmer et~al., 2017]{grimmeretal2017}
Grimmer, J., Messing, S., \& Westwood, S.~J. (2017).
\newblock Estimating heterogeneous treatment effects and the effects of
  heterogeneous treatments with ensemble methods.
\newblock {\em Political Analysis}, 25(4), 413--434.

\bibitem[Haim et~al., 2021]{haimetal2021}
Haim, D., Ravanilla, N., \& Sexton, R. (2021).
\newblock Sustained government engagement improves subsequent pandeic risk
  reporting in conflict zones.
\newblock {\em {American Political Science Review}}, 115(2), 717--724.

\bibitem[Hainmueller et~al., 2018]{hainmuelleretal20}
Hainmueller, J., Mummolo, J., \& Xu, Y. (2018).
\newblock How much should we trust estimates from multiplicative interaction
  models? simple tools to improve empirical practice.
\newblock {\em Political Analysis}, 27(2), 163--192.

\bibitem[Holland, 1986]{holland1986}
Holland, P.~W. (1986).
\newblock Statistics and causal inference.
\newblock {\em Journal of the American Statistical Association}, 81(396),
  945--960.

\bibitem[Huang, 2024]{huang2024sensitivity}
Huang, M.~Y. (2024).
\newblock Sensitivity analysis for the generalization of experimental results.
\newblock {\em Journal of the Royal Statistical Society Series A: Statistics in
  Society}, 187(4), 900--918.

\bibitem[Imai et~al., 2010a]{imai2010general}
Imai, K., Keele, L., \& Tingley, D. (2010a).
\newblock A general approach to causal mediation analysis.
\newblock {\em Psychological methods}, 15(4), 309.

\bibitem[Imai et~al., 2011]{imai2011unpacking}
Imai, K., Keele, L., Tingley, D., \& Yamamoto, T. (2011).
\newblock Unpacking the black box of causality: Learning about causal
  mechanisms from experimental and observational studies.
\newblock {\em American Political Science Review}, 105(4), 765--789.

\bibitem[Imai et~al., 2010b]{imai2010identification}
Imai, K., Keele, L., \& Yamamoto, T. (2010b).
\newblock Identification, inference and sensitivity analysis for causal
  mediation effects.
\newblock {\em Statistical science}, 25(1), 51--71.

\bibitem[Imai \& Yamamoto, 2013]{imai2013identification}
Imai, K. \& Yamamoto, T. (2013).
\newblock Identification and sensitivity analysis for multiple causal
  mechanisms: Revisiting evidence from framing experiments.
\newblock {\em Political Analysis}, 21(2), 141--171.

\bibitem[Incerti, 2020]{incerti2019}
Incerti, T. (2020).
\newblock Corruption information and vote share: A meta-analysis and lessons
  for experimental design.
\newblock {\em {American Political Science Review}}, 114(3), 761--774.

\bibitem[Kitagawa \& Tetenov, 2018]{kitagawatetenov2018}
Kitagawa, T. \& Tetenov, A. (2018).
\newblock Who should be treated? empirical welfare maximization methods for
  treatment choice.
\newblock {\em Econometrica}, 86(2), 591--616.

\bibitem[Lee \& Shaikh, 2014]{leeshaikh2014}
Lee, S. \& Shaikh, A.~M. (2014).
\newblock Multiple testing and heterogeneous treatment effects: Re-evaluating
  the effect of progresa on school enrollment.
\newblock {\em Journal of Applied Econometrics}, 29, 612--626.

\bibitem[Little et~al., 2022]{littleetal2022}
Little, A.~T., Schnakenberg, K.~E., \& Turner, I.~R. (2022).
\newblock Motivated reasoning and democratic accountability.
\newblock {\em {American Political Science Review}}, 116(2), 751--767.

\bibitem[Manski, 1997]{manski1997}
Manski, C.~F. (1997).
\newblock Monotone treatment response.
\newblock {\em Econometrica}, 65(6), 1311--1334.

\bibitem[McClelland \& Judd, 1993]{mcclellandjudd376}
McClelland, G.~H. \& Judd, C.~M. (1993).
\newblock Statistical difficulties of detecting interactions and moderator
  effects.
\newblock {\em Psychological Bulletin}, 114(2), 376--390.

\bibitem[Moscowitz, 2021]{moscowitz2021}
Moscowitz, D. (2021).
\newblock Local news, information, and the nationalization of u.s. elections.
\newblock {\em {American Political Science Review}}, 115(1), 114--129.

\bibitem[Neyman, 1923]{neyman1923applications}
Neyman, J. (1923).
\newblock Sur les applications de la theorie des probabilites aux experiences
  agricoles: essai des principes (masters thesis); justification of
  applications of the calculus of probabilities to the solutions of certain
  questions in agricultural experimentation. excerpts english translation
  (reprinted).
\newblock {\em Statistical Science}, 5, 463--472.

\bibitem[Nilsson et~al., 2021]{nilssonetal2021}
Nilsson, A., Bonander, C., Str\"{o}mberg, U., \& Bj\"{o}rk, J. (2021).
\newblock A directed acyclic graph for interactions.
\newblock {\em International Journal of Epidemiology}, 50(2), 613--619.

\bibitem[Rubin, 1974]{rubin1974estimating}
Rubin, D.~B. (1974).
\newblock Estimating causal effects of treatments in randomized and
  nonrandomized studies.
\newblock {\em Journal of educational Psychology}, 66(5), 688.

\bibitem[Slough, 2023]{slough2022phantom}
Slough, T. (2023).
\newblock Phantom counterfactuals.
\newblock {\em American Journal of Political Science}, 67(1), 137--153.

\bibitem[Slough, 2024]{slough2024}
Slough, T. (2024).
\newblock Bureaucratic quality and electoral accountability.
\newblock {\em {American Political Science Review}}, 118(4), 1931--1950.

\bibitem[Slough \& Tyson, 2023]{sloughtyson2023ajps}
Slough, T. \& Tyson, S.~A. (2023).
\newblock External validity and meta-analysis.
\newblock {\em American Journal of Political Science}, 67(2), 440--455.

\bibitem[Slough \& Tyson, 2024]{sloughtyson2023}
Slough, T. \& Tyson, S.~A. (2024).
\newblock {\em External Validity and Evidence Accumulation}.
\newblock New York: {Cambridge University Press}.

\bibitem[Slough \& Tyson, 2025]{slough2025sign}
Slough, T. \& Tyson, S.~A. (2025).
\newblock Sign-congruence, external validity, and replication.
\newblock {\em Political Analysis}, 33(3), 195--210.

\bibitem[Weinberg, 2007]{weinberg2007}
Weinberg, C.~R. (2007).
\newblock Can dags clarify effect moderation?
\newblock {\em Epidemiology}, 18(5), 569--572.

\end{thebibliography}

\end{document}